\documentclass[conference]{IEEEtran}

\usepackage{cite}

\usepackage{graphicx}
\usepackage{amsmath}
\usepackage[ruled,vlined,linesnumbered]{algorithm2e}

\usepackage{array}

\ifCLASSOPTIONcompsoc
 \usepackage[caption=false,font=normalsize,labelfont=sf,textfont=sf]{subfig}
\else
 \usepackage[caption=false,font=footnotesize]{subfig}
\fi

\usepackage{stfloats}
\usepackage{xurl}

\usepackage{amssymb}
\usepackage{amsthm}
\usepackage{subcaption}
\usepackage{xcolor}
\usepackage{hyperref}
\usepackage{xspace}
\usepackage{multirow}
\usepackage{multicol}
\usepackage{booktabs}
\usepackage[table]{xcolor}
\usepackage{colortbl}
\usepackage{diagbox}
\usepackage{enumitem}
\usepackage{makecell}
\usepackage{array} 
\usepackage{amssymb}   
\usepackage{pifont}    
\newcommand{\cmark}{\ding{51}}%
\newcommand{\xmark}{\ding{55}}%
\renewcommand{\sectionautorefname}{Section}

\def\C{\textit{Chameleon}\xspace}
\def\eg{\textit{e.g.}\xspace}

\newtheorem{theorem}{Theorem}
\newtheorem{lemma}{Lemma}

\begin{document}
%
\title{Chameleon: Robust Defense Against Tor Website Fingerprinting via Many-to-Many Traffic Morphing}





\author{
\IEEEauthorblockN{
Yuwen Cui\IEEEauthorrefmark{1},
Kai Wei\IEEEauthorrefmark{1},
Kehan Shen\IEEEauthorrefmark{1},
Ning Wang\IEEEauthorrefmark{1},
Zhuo Lu\IEEEauthorrefmark{2},
Yao Liu\IEEEauthorrefmark{1},
Guangjing Wang\IEEEauthorrefmark{1}
}

\IEEEauthorblockA{
\IEEEauthorrefmark{1}
Bellini College of Artificial Intelligence, Cybersecurity and Computing, University of South Florida
}

\IEEEauthorblockA{
\IEEEauthorrefmark{2}
Department of Electrical and Computer Engineering, University of South Florida
}
}
	

%


\IEEEoverridecommandlockouts
\makeatletter\def\@IEEEpubidpullup{6.5\baselineskip}\makeatother

\maketitle

\begin{abstract}
Website fingerprinting (WF) attacks can infer users' browsing activities from encrypted Tor traffic by exploiting side-channel features. Although many WF defenses have been proposed, we find that most existing defenses create learnable web trace mapping features. We further show that robustness against adversarial training does not necessarily imply robustness against defense-aware autoencoder (DAAE)-based attacks.

To address these limitations, we present \C, a robust WF defense based on many-to-many randomized traffic morphing. \C selects morphing candidates with high intra-class diversity and low inter-class disparity. \C randomly maps each webpage trace to multiple candidates, and allows different webpages to share morphing targets, thereby increasing adversarial uncertainty. For practical Tor deployment, \C introduces a radix-trie-based synchronization mechanism that enables pluggable transport (PT) endpoints to identify consistent morphing traces using packet-direction prefixes, together with trace mutation and normalized prefix matching to reduce overhead. We evaluate \C against six state-of-the-art defenses and five WF attacks on three public datasets in closed- and open-world settings. Compared with Adaptive Tamaraw, \C reduces adversarial-training-based attack accuracy by up to 36.74\% while reducing bandwidth and time overhead by 34.12\% and 60.38\%, respectively. Under DAAE-based RF attacks on GTT23, \C limits attack performance to 35.19\% F1-score while Adaptive Tamaraw only limits it to 88.22\% F1-score. In the real-world PT bridge evaluation, \C substantially reduces the effectiveness of strong WF attacks while incurring only 16.25\% time overhead.

\end{abstract}


%
\IEEEpeerreviewmaketitle

\section{Introduction}
\label{sec:intro}

Encrypted networks such as Tor~\cite{goldschlag1999onion} remain vulnerable to website fingerprinting (WF) attacks~\cite{shen2023subverting, bahramali2023realistic, sirinam2018deep, bhat2018var}. 
For example, observable side-channel information in Tor network traffic, such as traffic timing~\cite{wang2016realistically, dyer2012peek, wang2014effective, kwon2015circuit}, traffic bursts~\cite{zhao2024towards, bahramali2023realistic}, and traffic direction~\cite{rahman2019tik, bhat2018var}, can still reveal web browsing activities.

Many defense mechanisms~\cite{cai2014systematic, juarez2016toward, wang2017walkie} have been proposed to defend against WF attacks. For example, adaptive padding-based obfuscation techniques~\cite{panchenko2011website, juarez2016toward, gong2020zero} insert dummy packets to obfuscate features, such as packet timing and traffic bursts. 
Adversarial perturbation techniques~\cite{li2022minipatch, xie2025gapdis, qiao2024trace, wang2026cease} generate adversarial noise to misguide deep learning-based WF attacks.
Regularization-based methods~\cite{wang2014effective, holland2020regulator, wang2017walkie} transform traffic traces into regularized sequences with similar patterns, reducing the distinguishability among different website traces.

\begin{figure}[t]
  \centering
  \includegraphics[scale=0.32]{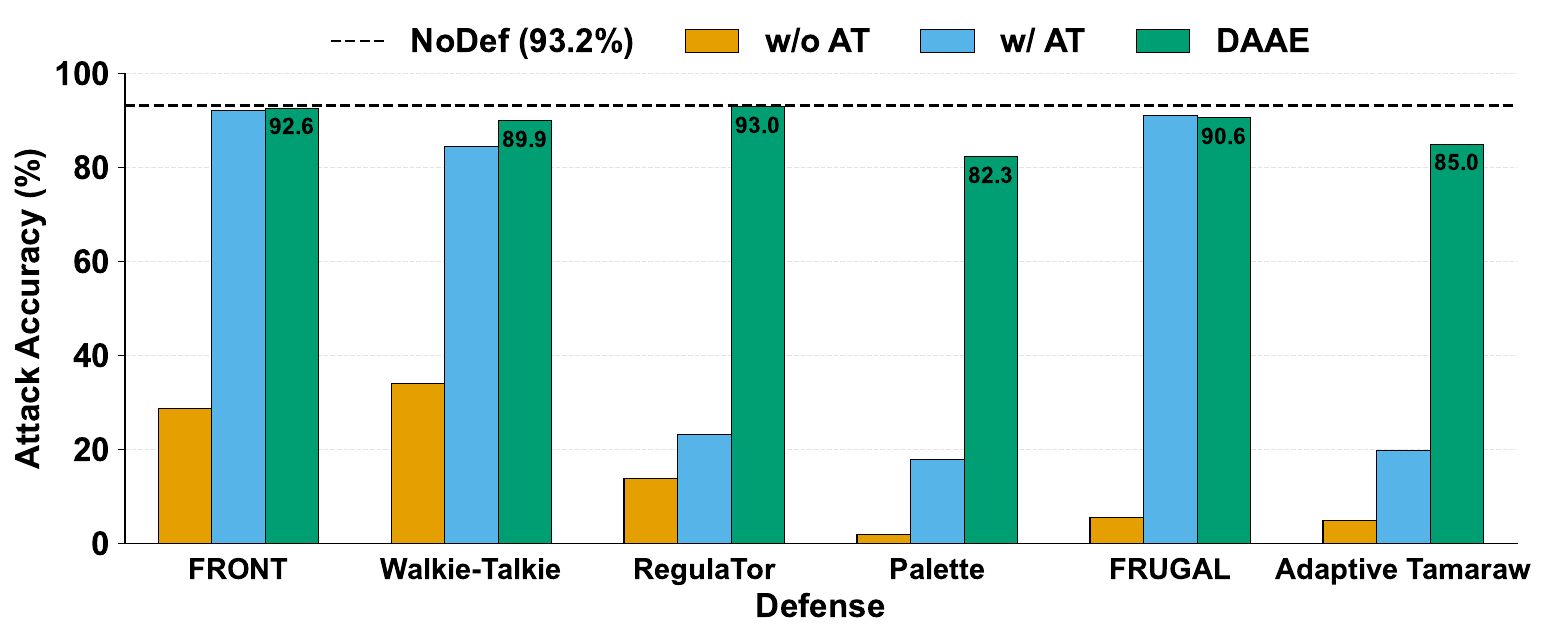}
  \caption{An exemplar of WF defense evaluation on the GTT23~\cite{jansen2026measurement} dataset under adversarial training (AT)-assisted and DAAE-assisted deep fingerprinting (DF)~\cite{sirinam2018deep} attack. The legend includes non-defended traffic (NoDef),
  traffic defended with different defenses against the DF attack without adversarial training (w/o AT), against the DF attack with adversarial training (w/ AT), and against the DAAE-assisted DF (DAAE).}
  \label{fig:introcom}
  \vspace{-17pt}
\end{figure}

However, we discover that most existing defenses create stable, learnable defense-induced mappings, which can be utilized by \textit{defense-aware} attacks.
Specifically, we assume a strong WF adversary who is aware that a defense is deployed on Tor but does not know the specific defense mechanism being used.
The adversary can train a Defense-Aware Autoencoder (DAAE) to learn defense-induced perturbation patterns using paired original Tor website traces and defense-perturbed traces.
The specific design of DAAE is in \appendixautorefname~\ref{subsec:daae}.
We find that DAAE can easily recover raw traffic features from previously unseen defended traces.
As shown in \figureautorefname~\ref{fig:introcom}, existing defenses reduce the website fingerprinting attack model DF~\cite{sirinam2018deep} accuracy from 93.2\% to 2\%-35\%.
However, the DAAE-based DF attack recovers attack performance against these defenses, achieving over 90\% accuracy against RegulaTor~\cite{holland2020regulator} and over 80\% against Palette~\cite{shen2024real} and Adaptive Tamaraw~\cite{khajavi2026lightening}.
Similarly, adversarial training mostly recovers the attack performance against several defenses, with accuracy exceeding 90\% for FRONT~\cite{gong2020zero} and FRUGAL~\cite{wang2026cease} and 80\% for Walkie-Talkie~\cite{wang2017walkie}.
Furthermore, we reveal that \textbf{defense robustness against adversarial training does not necessarily imply robustness to DAAE-based attacks}. For instance, RegulaTor, Palette, and Adaptive Tamaraw remain resilient to adversarial training-assisted attacks, but can be compromised by DAAE-based attacks.

In this work, we design \C to eliminate stable mappings through randomized many-to-many transformations for a robust WF defense. 
\C is based on traffic morphing, which transforms the traffic pattern of a given website to resemble that of another site~\cite{wang2014effective, al2019bimorphing}. 
Through a many-to-many morphing design, \C randomly maps a single website trace to multiple target morphing traces. Each target trace can correspond to multiple source traces.
The key innovation is a novel design to ensure consistency between the pluggable transport client and server (PTs)~\cite{torptspec2022, gong2023wfdefproxy} in real-world deployments of many-to-many random morphing.
\C theoretically increases the uncertainty faced by the adversary and thereby raises the difficulty of DAAE and adversarial training-based WF attacks. 
In addition, we design a radix trie–based structure and trace mutation algorithm that enables efficient morphing target selection, so as to reduce the overhead when dealing with unseen websites. Specifically, we address the following three challenges.

\textit{First, how to defend against defense-aware adversarial training-based and DAAE-based attacks?} To obfuscate defense-aware attacks, we design methods to collect traffic that exhibits high intra-class diversity and low inter-class disparity, which is used as traffic morphing candidates. During online deployment, each webpage trace is morphed to one of these candidates, misleading WF attack models. Specifically, we design a lightweight offline trace-selection algorithm that characterizes the geometric relationships among traffic traces and identifies representative traces with high intra-class diversity and low inter-class disparity. 
In addition, \C randomly maps each webpage trace to multiple candidate traces while allowing different webpages to share the same morphing targets.
This randomized mapping eliminates the deterministic traffic transformation exploited by DAAE, thereby substantially improving defense robustness.

\begin{figure}[t]
  \centering
  \includegraphics[scale=0.28]{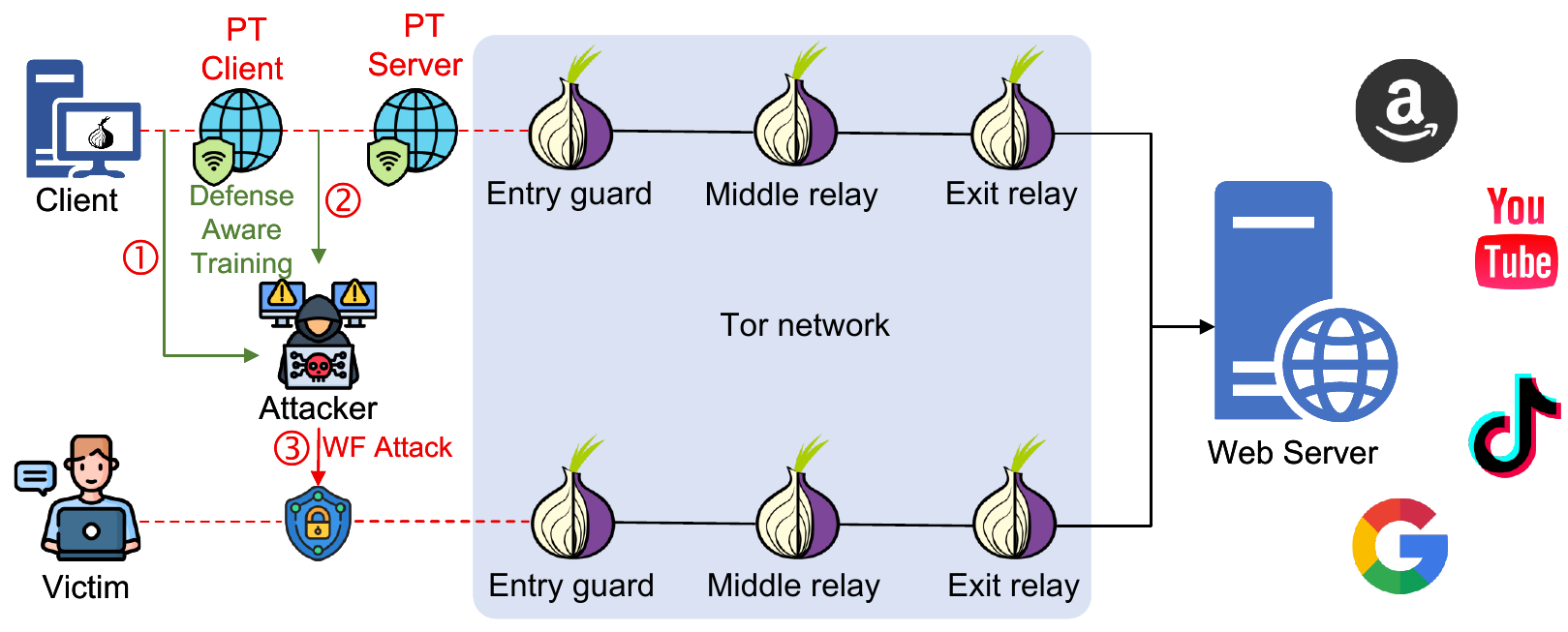}
  \caption{Overview of the Tor network and Threat Position.}
  \label{fig:torarch}
  \vspace{-17pt}
\end{figure}

\textit{Second, how can a many-to-many traffic morphing strategy be synchronized between the PT client and the PT server?} The defenses~\cite{gong2020zero, holland2020regulator, xie2025gapdis, shen2024real}, including \C, can be deployed through PTs, as shown in \figureautorefname~\ref{fig:torarch}. A PT client and PT server provide user-configurable traffic defenses without modifying the underlying Tor protocol~\cite{torptspec2022}. However, synchronizing web traffic between the PT client and bridge after many-to-many random morphing is challenging because the randomization can induce the endpoints to select inconsistent morphing traces. 

We observe that packet direction sequences are determined from the packet departure timestamps at the PT endpoint during real-time transmission rather than their delivery timestamps. This observation provides a consistent basis for synchronization between the PT client and PT server.
Therefore, we design a radix trie~\cite{krishnaraj2021efficient}-based structure to pre-match candidate morphing traces according to their packet-direction prefixes.
The PT endpoint that first determines the morphing trace can send a short sequence of consecutive packet directions from the selected trace as an implicit synchronization signal. Upon receiving this sequence, the peer PT endpoint performs a prefix search on its local radix trie and identifies the same morphing trace without requiring an explicit trace identifier. Once the trace is identified, both endpoints apply the same morphing pattern to subsequent traffic, thereby maintaining consistent real-time morphing across the PTs.

\textit{Third, how to reduce the extra time and bandwidth overhead incurred by the morphing process?} As the number of candidate traces grows, the morphing process introduces additional computation time and bandwidth overhead. 
The core question is how to achieve a trade-off between the overhead and the defense performance. We have two main designs. (i) We propose a trace mutation algorithm that minimizes additional padding when computing the distance between the real-time input packets and the selected target morphing trace. In this way, we can reduce the overhead introduced by the potential trace difference from unseen webpages. (ii) We design a row-wise normalization scheme for the candidate traces to enable efficient trace matching. The normalization removes scale and offset variations within each trace, ensuring that matching decisions are driven by intrinsic traffic patterns rather than absolute magnitudes. Then, we construct prefix-based feature vectors, where traces sharing common structural prefixes are organized hierarchically within the radix trie. This design facilitates $log$-time lookup during morphing target selection. 

We compare \C with six state-of-the-art (SOTA) WF defenses against five SOTA WF attack models on three widely used public datasets in the closed- and open-world scenarios. 
Compared with the recent Adaptive Tamaraw defense~\cite{khajavi2026lightening}, \C further reduces
adversarial training-based WF attack accuracy by up to 36.74\%, while also reducing bandwidth overhead by 34.12\% and time overhead by 60.38\% on the DF dataset.
In particular, under DAAE-based RF attacks~\cite{shen2023subverting}, \C limits attack performance on the GTT23 dataset to 40.68\% precision and a 35.19\% F1-score. In comparison, the same RF attack achieves 80.64\% precision and a 70.30\% F1-score under Palette defense, and 89.82\% precision and an 88.22\% F1-score under Adaptive Tamaraw defense strategy.

Overall, our main contributions are listed as follows:
\begin{itemize}
    \item We design a robust WF defense model that builds candidate website traces with high intra-class diversity and low inter-class disparity to remove conspicuous features in morphing.
    \item We introduce a radix trie structure to efficiently match a morphing trace from the website trace candidates and a trace mutation algorithm to efficiently reduce overhead.
    \item We evaluate SOTA WF defenses against adversarial training-based and DAAE-based attacks, revealing limitations in existing WF defense mechanisms. 
    \item We conduct a comprehensive evaluation and demonstrate that \C exhibits superior robustness to adversarial training-based WF attacks with modest overhead.
\end{itemize}

\section{Background and Motivation}
\label{sec:background}


\begin{figure}[t]
  \centering
  \includegraphics[scale=0.40]{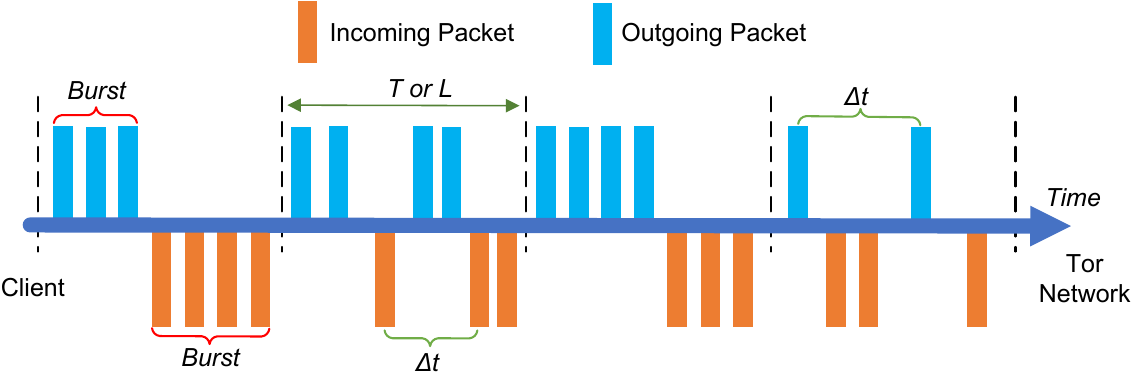}
  \caption{Tor traffic timing trace with direction, bursts, and ordering from client to Tor network.}
  \label{fig:timingtrace}
  \vspace{-10pt} 
\end{figure}

\subsection{Tor Traffic Features}
The main Tor traffic features, as shown in \figurename~\ref{fig:timingtrace}, include packet inter-arrival time, burst, packet ordering, direction, and so on. Specifically, traffic direction is encoded as +1 for outgoing packets from the client to the Tor network and -1 for incoming packets from the Tor network to the client~\cite{rahman2019tik, shi2025multiscale, zhao2024towards}. The traffic direction provides packet ordering information, which describes the sequence in which packets are transmitted and received during a browsing session. The packet ordering captures the precise arrangement of request-response interactions between a client and a server~\cite{wang2013improved}.
The traffic bursts are sequences of consecutive packets in the same direction, either incoming or outgoing. The burst patterns capture higher-level structural information about how websites load resources, such as videos, audio, and pictures, reflecting grouped transactions between browsers and servers. 

Packet inter-arrival time~\cite{wang2016realistically} measures the time interval $\Delta t$ between consecutive packets in the same direction. The $T$ and $L$, as shown in \figurename~\ref{fig:timingtrace}, denote the fixed time window and the traffic length for captured Tor flows, respectively. Generally, the Tor protocol splits traffic into cells with a fixed size of 512 bytes. The fixed time window $T$ or the traffic length $L$ helps standardize traffic traces into uniform representations. In this work, \C obfuscates trace-level features, including inter-packet timing and traffic burst structures, by morphing traces toward pre-selected candidate website traces. These features are specifically targeted, as they constitute the primary discriminative patterns exploited by SOTA WF attacks~\cite{bahramali2023realistic, shen2023subverting}.

\subsection{Tor Pluggable Transport Bridges}
Pluggable Transport (PT) bridges act as programmable proxies between the Tor client and entry guard.
In a PT deployment, the Tor client forwards traffic to a locally running PT client, which transforms the traffic before transmitting it to a corresponding PT server associated with an entry guard. The PT server reverses the transformation and forwards the resulting traffic to the Tor network. In this way, PT helps establish an end-to-end connection while hiding or modifying traffic characteristics observable on the client-bridge link.

The PT architecture provides an extensible interface through which traffic-obfuscation mechanisms can be deployed~\cite{torptspec2022}. A defense can be implemented within the PT client and server components and applied to traffic without requiring prior knowledge of the destination website or changes to the underlying Tor network. In particular, a client-side defense can delay packets and inject dummy packets before traffic leaves the PT client/server, while the corresponding bridge-side component coordinates any state required to reconstruct the original Tor communication. In summary, PTs offer a deployable and modular interface for evaluating WF defenses under realistic network conditions while preserving compatibility with the existing Tor infrastructure.

\subsection{Related Defense Work Comparison}
\label{subsec:motivation}

\begin{table}
\centering
\scriptsize
\setlength{\tabcolsep}{1pt}
\renewcommand{\arraystretch}{1.2}
\caption{A comparison of the WF defenses.}
\label{tab:relatedwork}
\begin{tabular}{c|c|ccccc}
\toprule
\textbf{Categories} & \textbf{Defenses} & \textbf{\makecell{Live \\ Traffic}} & \textbf{\makecell{AT \\ Resistance}} & \textbf{\makecell{DAAE \\ Resistance}}  & \textbf{\makecell{Supports \\ PT}}  \\

\midrule

\multirow{5}{*}{\textbf{Obfuscation}} & WTF-PAD~\cite{juarez2016toward} & \cmark & \xmark & -- & \cmark \\
& FRONT~\cite{gong2020zero} & \cmark & \xmark & -- & \cmark \\
& ALERT~\cite{qiao2024trace} & \xmark & \xmark & -- & \xmark \\
& GAPDiS~\cite{xie2025gapdis} & \cmark & \xmark & -- & \cmark \\
& FRUGAL~\cite{wang2026cease} & \cmark & \xmark & -- & \xmark \\

\rowcolor{gray!15}
 & CS-BuFLO~\cite{cai2014cs} & \xmark & \xmark & -- & \xmark \\
 \rowcolor{gray!15}
& Super-sequence~\cite{wang2014effective} & \xmark & \xmark & -- & \xmark \\
\rowcolor{gray!15}
& Palette~\cite{shen2024real} & \cmark & \cmark & \xmark & \cmark \\
\rowcolor{gray!15}
 & Adaptive Tamaraw~\cite{khajavi2026lightening} & $\circ$ & \cmark & \xmark & \xmark \\
\rowcolor{gray!15}
\multirow{-5}{*}{\textbf{Regularization}} & \textbf{\C(Ours)} & \textbf{\cmark}  & \textbf{\cmark} & \textbf{\cmark} & \cmark \\
\bottomrule

\end{tabular}
\begin{minipage}{0.98\linewidth}

\rule{0pt}{2pt}
\setlength{\topskip}{5pt}
\setlength{\leftskip}{4pt}

$\circ$: Insufficient evidence from the paper and provided artifact; --: If the defense does not resist AT, the DAAE-based attack is not performed; Live Traffic and Supports PT are determined from the original papers and publicly available artifacts; AT and DAAE attack resistance are determined from our evaluations unless otherwise noted.
\end{minipage}
\vspace{-17pt}
\end{table}

We compare and summarize major defense methods in \tableautorefname~\ref{tab:relatedwork}. ``Live Traffic" denotes applicability to real-time traffic. ``AdvTrain Resistance" denotes adversarial training (AT)-based attack accuracy below 50\%. ``DAAE Resistance" denotes DAAE-based attack accuracy below 50\%. ``Supports PT" denotes that a PT-based pluggable defense implementation is provided. In particular, defenses that support live traffic adaptation, such as FRONT, GAPDiS, and FRUGAL, remain vulnerable to AT-based attacks. In contrast, defenses with stronger robustness, such as Palette, introduce additional constraints on traffic shaping and deployment. More importantly, existing defenses provide limited protection against DAAE-based attacks, which can explicitly learn the characteristics introduced by a deployed defense. A practical WF defense needs to jointly support live traffic, robustness against defense-aware attacks, and be deployable through a PT interface.

\section{Preliminary}
\label{sec:preliminary}


\subsection{Threat Model}
\label{subsec:threat}

\subsubsection{Attack Scenario}
We consider the standard local passive adversary that monitors the communication network traffic between the Tor client and the entry guard but cannot modify traffic packets, as shown in \figureautorefname~\ref{fig:torarch}.
The defense mechanisms can be deployed through a Pluggable Transport (PT) bridge~\cite{gong2023wfdefproxy, xie2025gapdis, shen2024real}, which acts as proxies between the Tor client and the entry guard to obfuscate traffic.

\subsubsection{Attacker Capability}
We consider a strong adversary with access to defended traffic traces from monitored webpages collected under conditions comparable to those experienced by a normal user.
We assume the adversary can collect both original Tor traffic at the client side (\textcircled{1} in \figureautorefname~\ref{fig:torarch}) and the corresponding defended traffic transmitted through the PT client and PT server (\textcircled{2} in \figureautorefname~\ref{fig:torarch}).

We consider a strong adversary who uses the collected traces to train either \emph{adversarial training-based WF attacks} or DAAE-based WF attacks. These two attack strategies represent increasing knowledge of the defense and enable us to evaluate whether a defense remains effective when the adversary explicitly adapts to traffic transformations for WF attacks.

\subsubsection{Defense Setting}
The \C does not require prior knowledge of the destination webpage. Thus, \C can be applied to arbitrary browsing webpages rather than a pre-configured set of static candidate traces.
Additionally, PTs are independently deployed by users. Therefore, the defense for different users does not rely on the same predefined payloads or destination-specific configurations. In this setting, \C can avoid introducing stable, webpage-dependent artifacts that an adaptive adversary can learn.

\subsection{Problem Formulation}
\label{sec:threatmodel}

The Tor network traffic is represented as a time-ordered sequence of packets. Formally, each traffic trace corresponding to the loading of a webpage is represented as: $w = [(t_{1}, d_{1}), (t_{2}, d_{2}), \ldots, (t_{N}, d_{N})]$. $N$ is the total number of website packets observed during the session. $t_{i} \in \{t_{1}, t_{2}, \dots, t_{N}\}$ indicates the timestamp at the $i$-th packet in the trace sequence. $d_{i} \in \{+1, -1\}$ denotes the direction of the $i$-th packet.

\begin{figure*}[t]
  \centering
  \includegraphics[scale=0.69]{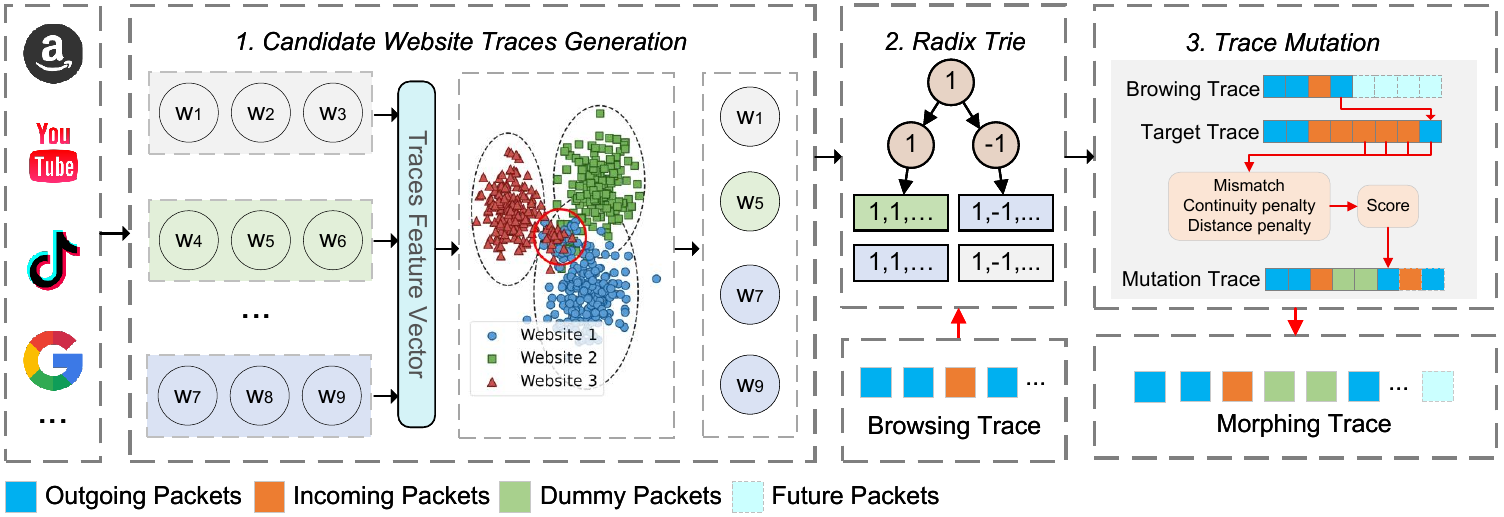}
  \caption{Overview \C architecture. 1) Separate the high intra-class diversity and low inter-class disparity data as morphing trace candidates. 2) Construct a radix trie based on the selected trace set. 3) Apply trace mutation to reconcile burst discrepancies between real-time traffic and target morphed traces.}
  \label{fig:chameleon}
  \vspace{-15pt}
\end{figure*}

Let $\mathcal{W}$ denote the monitored websites and $\mathcal{X}$ denote the space of observable traffic. 
An adversary aims to learn a classification model $h:\mathcal{X} \rightarrow \mathcal{W}$ that predicts the specific website from an observed trace. The WF defense system can be regarded as a stochastic transformation mechanism $\mathcal{D}: \mathcal{W} \rightarrow \mathcal{D}(\mathcal{X})$, where $\mathcal{D}(\mathcal{X})$ denotes a distribution over defended traces.
$\mathcal{R} \subseteq \mathcal{W} \times \mathcal{W}^{'}$ denotes the defense relation, where $(w,w^{'}) \in \mathcal{R}$ indicates that $w^{'}$ is a valid defended trace generated from the original trace $w$ by the defense mechanism.
A WF defense usually modifies the $w$ to $w^{'}$ by injecting dummy packets $d^{'}$ or increasing the timing delay $\tau$. 
In \C, we model the defense mechanism as a mapping between original traffic traces and defended traffic traces: $f: W \to W^{'}$. $W$ represents the set of all collected traffic traces. $W^{'}$ is the set of all manipulated traces derived from $W$. Formally, the $f: W \to W^{'}$ satisties:
\begin{equation}
\begin{aligned}
\forall w \in W,\ \exists w_1^{'}, w_2^{'} \in W^{'},\ w_1^{'} \neq w_2^{'},\\ \text{s.t. } (w, w_1^{'}) \in \mathcal{R},\ (w, w_2^{'}) \in \mathcal{R}.
\end{aligned}
\end{equation}
The many-to-many mapping enables multiple transformations per trace and partial overlap across classes, thereby disrupting both intra-class similarity and inter-class disparity. We provide a theoretical analysis of attacker success by analyzing the many-to-many mapping used by \C in \appendixautorefname~\ref{subsec:boundproof}.

Following the literature~\cite{shen2024real, khajavi2026lightening}, the bandwidth overhead of defense is measured as the ratio of the total amount of dummy data to the total amount of real packets, denoted as $B(w, w^{'})= \frac{|w^{'}|-|w|}{|w|}$. Similarly, the time overhead is measured as $T(w, w^{'})= \frac{t_{k}-t_{N}}{t_{N}}$, where $t_{k}$ is the time spent transmitting the defended trace $w^{'}$.

\section{System Design}
\label{sec:design}

\subsection{Candidate Website Traces Generation}
\label{subsec:set}

We observe that the existing WF attacks achieve high classification accuracy by learning feature representations with high intra-class similarity and high inter-class disparity. 
Therefore, to compromise the WF attacks, we need to obfuscate the intra-class and inter-class features. 
Specifically, we select samples that are more scattered inside the same class (high intra-class diversity) and closer to neighboring classes (low inter-class disparity). 
High intra-class diversity expands the apparent pattern of each class. If selected website traces from the same website vary significantly, the WF attack model needs to generalize over a wider region.
Meanwhile, low inter-class disparity reduces the margin between classes: traces from different websites become harder to separate. 
The two properties make learned decision boundaries less stable and increase confusion around class borders.

To depict trace candidates that exhibit high intra-class diversity and low inter-class disparity, we model traces with direction and timing features using a distance-based method.
The advantage is that \C can capture relationships among neighboring samples (\eg, from the same website) with minimal modeling assumptions. Additionally, distance-based neighborhood modeling remains robust when class manifolds are irregular or non-linear, which is suitable for extracting complex patterns from monitored website traffic traces.

\begin{algorithm}[t]
\caption{HiLoD algorithm}
\label{alg:traces_selection}

\KwIn{
$\textit{traces}$,
$\textit{labels}$,
neighbor size $k$,
selection ratio $\rho$
}
\KwOut{
$\textit{selected\_traces}$,
$\textit{selected\_trace\_labels}$}

$n \gets |\textit{traces}|$\;
\For{$i \gets 1$ \KwTo $n$}{
    $\mathbf{X}[i] \gets \textsc{TraceToFeature}(\textit{traces}[i])$\; 
}

$\boldsymbol{\mu} \gets \text{mean}(\mathbf{X}, \text{axis}=0)$, $\boldsymbol{\sigma} \gets \text{std}(\mathbf{X}, \text{axis}=0) + 10^{-6}$\; 
$\mathbf{X}_n \gets (\mathbf{X} - \boldsymbol{\mu}) / \boldsymbol{\sigma}$, $k_{\text{eff}} \gets \min(\max(3,k),\, n-1)$\;
$(\mathbf{dists}, \mathbf{nbrs}) \gets \textsc{Nearest Neighbors}(\mathbf{X}_n, k_{\text{eff}}+1)$; \\

Initialize $\textit{scores}[1..n] \gets 0$\;

\For{$i\gets1$ \KwTo $n$}{
    $\textit{nn\_idx}
    \gets \mathbf{nbrs}[i,2:k_{\mathrm{eff}}+1]$\;
    $\textit{nn\_dist}
    \gets \mathbf{dists}[i,2:k_{\mathrm{eff}}+1]$\;

    $\textit{nn\_lab}
    \gets \textit{labels}[\textit{nn\_idx}]$\;

    $\mathcal{S}_i \gets \{ j \in \{1,\ldots,k_{\text{eff}}\} : \textit{nn\_lab}[j] = \textit{labels}[i] \}$
    
    $\mathcal{D}_i \gets \{ j \in \{1,\ldots,k_{\text{eff}}\} : \textit{nn\_lab}[j] \neq \textit{labels}[i] \}$


    \eIf{$|\mathcal{S}_i|>0$}{

        $\textit{intra}_i \gets \operatorname{mean}(\textit{nn\_dist}[t] \mid t \in \mathcal{S}_i)$
    }{
        
        $\textit{intra}_i
        \gets \max(\textit{nn\_dist})$\;
    }

    \eIf{$|\mathcal{D}_i|>0$}{
    
        $\textit{inter}_i \gets \operatorname{mean}(\textit{nn\_dist}[t] \mid t \in \mathcal{D}_i)$
        
    }{
        $\textit{inter}_{i}
        \gets \max(\textit{nn\_dist})$\;
    }
    $\textit{mislead}_i \gets |\mathcal{D}_i| / k_{\text{eff}}$

    $\textit{scores}[i]
    \gets \textit{intra}_{i}
    - \textit{inter}_{i}
    + \textit{mislead}_{i}$\;
    
}

$m \gets \min(\lceil \rho * n \rceil, n)$\;

\Return \textit{selected\_website\_traces}, \textit{selected\_trace\_labels} $\gets$ Sorted Top-$m$ indices by descending $\textit{scores}$ \;

\end{algorithm}







In addition, a dedicated burst model is constructed from the direction sequence. \C detects direction change boundaries and computes lengths of consecutive same-direction packets. From these lengths, \C extracts burst mean, burst standard deviation, median, 90th percentile, burst count, and proportions of short bursts (\eg, shorter than 2) and long bursts (\eg, longer than 8). This provides a compact descriptor of patterns between inbound and outbound traffic, which are highly discriminative in Tor traffic analysis. 

The algorithm for generating high intra-class diversity and low inter-class disparity (HiLoD) website trace candidates is presented in \algorithmcfname~\ref{alg:traces_selection}.
To mitigate the potential information leakage from the morphing trace candidates, \C first converts each trace from the monitored dataset into a fixed-dimensional representation to support downstream distance-based modeling (lines 2-3 in \algorithmcfname~\ref{alg:traces_selection}). Specifically, \C decomposes the raw trace into three complementary information channels as the feature vector: traffic volume, temporal dynamics, and directional structure~\cite{li2018measuring}. Volume statistics include packet counts by direction (incoming and outgoing packets) and cumulative bytes by direction, as well as absolute-size descriptors (\eg, first 20 packets, last 30 packets, and standard deviation). Temporal dynamics are captured through duration ($t_{N}-t_{1}$) and inter-packet-time (IPT) summaries, including mean, standard deviation, quartiles, and extrema. This explicitly encodes both central tendency and dispersion of timing, which is important for distinguishing traffic classes that differ in burstiness or pacing.

Then, for each website trace, \C first estimates the average distance to same-class neighbors (lines 14-17 in \algorithmcfname~\ref{alg:traces_selection}). \C calculates the average distance to neighbors of different classes (lines 18-21). Finally, \C computes the mislead fraction of different-class neighbors (line 22). Line 23 defines the formula to calculate the score for prominent traces that are diverse within a class. In this formula, $intra$ denotes the average distance to same-class neighbors, where larger values indicate greater diversity within the class. Conversely, $inter$ represents the average distance to different-class neighbors, where smaller values suggest increased overlap across classes. A higher $score$ therefore corresponds to traces that are more ambiguous and harder for WF attack models to distinguish, making them less cleanly separable in the feature space. The candidate website traces are sorted by descending $score$ value.
\C selects the top-scoring subset that directly targets regions where classifiers are most error-prone. In line 24, parameter $\rho$ is the selection ratio, which controls the injectivity boundary for traffic morphing of selected website trace candidates.

\subsection{Radix Trie Construction}
\label{subsec:radix}

\begin{figure}[t]
  \centering
  \includegraphics[scale=0.42]{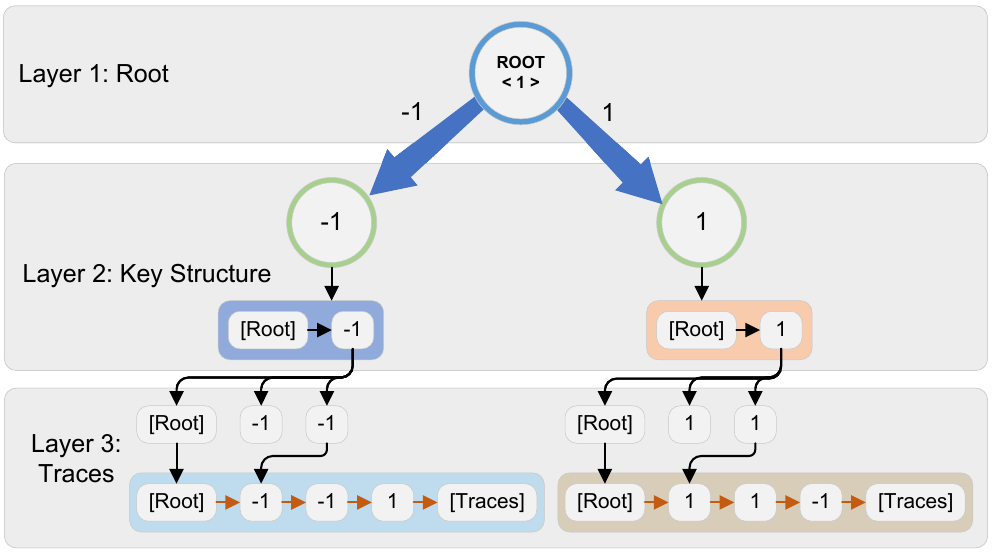}
  \caption{Visualization of Radix Trie. Traces are organized into a hierarchical key structure based on shared prefixes.}
  \label{fig:radix}
  \vspace{-15pt}
\end{figure}

With the generated candidate trace pool in Section~\ref{subsec:set}, we index the pool with a radix trie over binary direction sequences. As illustrated in \figureautorefname~\ref{fig:radix}, the radix trie is a compressed prefix-tree structure defined over directions $\{-1,+1\}$, where each root-to-node path corresponds to a direction prefix of the real-time browsing traffic. Following the traffic trace abstraction, each edge encodes a single outgoing or incoming packet step. At each node, we associate pool traces that share the same prefix, enabling sublinear candidate search rather than an exhaustive scan over variable-length traces.

However, prefix matching alone can be ambiguous: short prefixes may correspond to heterogeneous candidates, while long prefixes may overly constrain the search space. To regularize candidate selection, we construct a fixed-length embedding for each trace. Given a direction sequence $w = (d_{1}, \dots, d_{T})$ with $d_i \in \{-1,+1\}$, we truncate or zero-pad it to length $L$ to obtain $\mathbf{x} \in \mathbb{R}^L$. Then, we stack all embeddings to yield a matrix $\mathbf{X} \in \mathbb{R}^{N \times L}$. We then apply row-wise z-score normalization followed by $\ell_2$ normalization. The normalization design is described in \appendixautorefname~\ref{subsec:zscore}. To maintain coherence between the PT client and PT server in real-world deployment, the matching candidates need to exhibit distinct directions over the next few packets. The next few packets are used to verify short-term consistency between client and guard node traces by checking the sequence of packet directions. This constraint ensures that immediate temporal structure remains compatible, preventing abrupt inconsistencies during morphing.


During real-time transmission, packet direction sequences are determined from the packet departure timestamps at the PT endpoint rather than their delivery timestamps. Such a setting provides a consistent basis for synchronization between the PT client and PT server. The radix-trie search progressively narrows the candidate set as additional packet directions are observed, until less than, for example, 10 candidates remain. 
The endpoint that first determines the morphing trace then sends a short sequence of consecutive dummy packets as an implicit synchronization signal. Specifically, the number of dummy packets corresponds to the shortest prefix that uniquely distinguishes the selected trace from the remaining candidates. Upon receiving these dummy packets, the peer endpoint performs a prefix search on its local radix trie and identifies the same morphing trace without requiring an explicit trace identifier. Since each additional packet direction introduces another branching decision in the radix trie, the candidate space can decrease exponentially with the prefix length, allowing the selected trace to be uniquely identified using a bounded number of packet directions. 

Under the balancing branch implementation, when the candidate set contains at most 10 traces, the selected trace is typically uniquely identified by a prefix of 4-5 packet directions. Once the trace is identified, both endpoints apply the same morphing pattern to subsequent traffic, thereby maintaining consistent real-time morphing across the PT bridge. Meanwhile, a given real-world browsing trace will be mapped to different target traces across repeated visits. In this way, even when the same website is visited multiple times during adversarial training data collection, the resulting traces exhibit low inter-instance separability. Therefore, the attacker's ability to learn stable and discriminative patterns is compromised.

\subsection{Trace Morphing and Mutation}
\label{subsec:mutation}

\begin{figure}[t]
  \centering
  \includegraphics[scale=0.6]{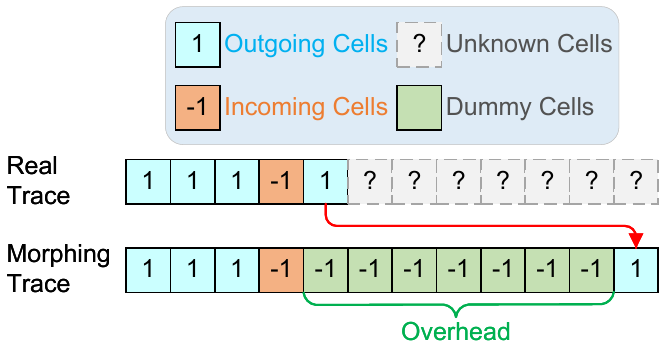}
  \caption{Visualization of Trace Morphing. Unknown packets correspond to unobserved traffic; Dummy packets are present in the morphing trace but absent from the real-time trace.}
  \label{fig:tracemorphing}
  \vspace{-14pt}
\end{figure}

After the radix trie construction, we can transform the real-time trace into a target morphing trace.
As illustrated in \figureautorefname~\ref{fig:tracemorphing}, the morphing trace is fully known in advance, whereas the real-time trace is only observable up to the time $t$, as future trace information remains unknown. 
Consequently, when the real-time captured packet deviates from the morphing trace in terms of packet sequence, the mismatch is compensated by inserting dummy packets derived from the morphing trace. However, these injected dummy packets introduce additional bandwidth consumption and may incur timing delays relative to the original real-time trace.

We categorize trace morphing into three categories: full-trace morphing, trace cut-off, and trace mutation. (i) In full-trace morphing, each real-time trace is fully morphed to match a selected target trace. If the trace completes while the target trace still contains a substantial number of remaining packets, the system continues transmitting dummy packets until the target trace reaches its end. Conversely, if the target trace terminates earlier while the real-time trace still contains pending packets, the remaining segment is reassigned to a newly selected target trace for continued transformation. Such a design maximizes indistinguishability at the cost of heavy overhead.

(ii) In the traffic trace cut-off scenario, the morphing process terminates when the real-time trace has been fully transformed, thereby avoiding the transmission of excessive dummy packets required to complete the target trace. Although this strategy reduces overhead, it may still introduce nontrivial padding and timing artifacts where there are long burst packets. As a result, the trace cut-off scenario introduces medium-level overhead.

(iii) We design the third morphing strategy, trace mutation, to further reduce overhead while preserving obfuscation effectiveness. As shown in \figureautorefname~\ref{fig:mutation}, when the target morphing trace exhibits a prolonged burst of packets whose direction differs from the real-time trace, the real-time trace will be temporarily buffered until a packet matches the direction in the target morphing trace.
Once such a match is encountered, the buffered packets can be released to the network. As mentioned in Section~\ref{subsec:radix}, morphing traces are randomly selected from a predefined morphing group. Consequently, the selected trace may not be the closest match to the underlying real-world traffic trace in terms of sequence similarity.
\C first performs a forward consistency check on the morphing trace and compares it with the real-time trace. Based on this comparison, \C identifies an appropriate position to flip exactly one direction bit in a short segment of the target morphing trace. After the flip, the target morphing trace's directions stay as close as possible to the corresponding real-time browsing directions.
In this way, we can reduce directional mismatch and mitigate buffering overhead. 

\begin{figure}[t]
  \centering
  \includegraphics[scale=0.6]{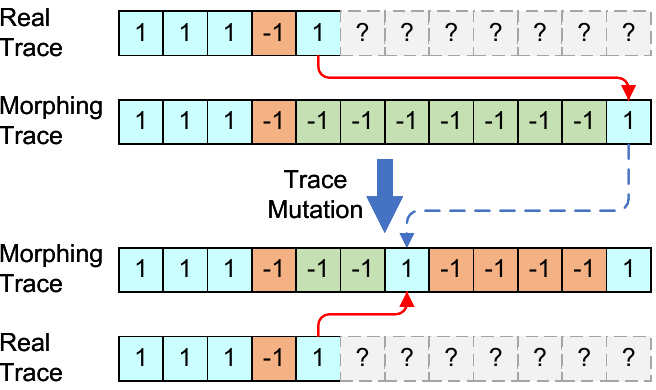}
  \caption{Visualization of Trace Mutation. }
  \label{fig:mutation}
  \vspace{-10pt}
\end{figure}

\begin{algorithm}[t]
\caption{Trace Mutation Algorithm}
\label{alg:tracemutation}
\KwIn{
$d^{target}$, 
$d^{in}$, 
$\ell$: the starting position of the trace mutation within the target morphing trace;
$j$: the real-time input packet location.
}
\KwOut{
$best\_loc$: best index selection for trace mutation.
}

\BlankLine
$n_{\mathrm{target}} \gets |d^{\mathrm{target}}|$, $n_{\mathrm{in}} \gets |d^{\mathrm{in}}|$\;

$idx_{\mathrm{left}} \gets \max(0,\ell)$\;
$idx_{\mathrm{right}} \gets \min(n_{\mathrm{target}}-1,\,\ell+M)$\;

\If{$idx_{\mathrm{left}} > idx_{\mathrm{right}}$}{
  \Return $\mathrm{clip}(\ell,\,0,\,n_{\mathrm{target}}-1)$\;
}

$idx_{\mathrm{w}} \gets \min(idx_{\mathrm{right}},\, n_{\mathrm{in}}-1)$\;

$W^{\mathrm{in}} \gets d^{\mathrm{in}}$, $W^{\mathrm{target}} \gets d^{\mathrm{target}}[idx_{\mathrm{left}} : idx_{\mathrm{w}}]$\;

$idx^{'} \gets idx_{\mathrm{left}}$, $score^{'} \gets +\infty$\;

\For{$i \gets idx_{\mathrm{left}}$ \KwTo $idx_{\mathrm{w}}$}{
  $W \gets \mathrm{copy}(W^{\mathrm{target}})$\;
  $W[i - idx_{\mathrm{left}}] \gets -W[i - idx_{\mathrm{left}}]$\;

  $\mathrm{mismatch} \gets \sum_{j} \mathbb{1}[W[j] \neq W^{\mathrm{in}}[j]]$\;

  $\mathrm{distance\_penalty} \gets |i - \ell|*0.25$\;

  $\mathrm{continuity\_penalty} \gets 0$\;
  \If{$i - 1 \ge 0$}{
    \eIf{$i - 1 \ge idx_{\mathrm{left}}$}{
      $\mathrm{continuity\_penalty} \gets (\mathrm{continuity\_penalty} + \mathbb{1}[W[i-1-idx_{\mathrm{left}}] \neq W[i-idx_{\mathrm{left}}]])/2$\;
    }{
      $\mathrm{continuity\_penalty} \gets (\mathrm{continuity\_penalty} + \mathbb{1}[d^{\mathrm{target}}[i-1] \neq W[i-idx_{\mathrm{left}}]])/2$\;
    }
  }
  \If{$i + 1 < n_{\mathrm{target}}$}{
    \eIf{$i + 1 \le idx_{\mathrm{w}}$}{
      $\mathrm{continuity\_penalty} \gets (\mathrm{continuity\_penalty} + \mathbb{1}[W[i+1-idx_{\mathrm{left}}] \neq W[i-idx_{\mathrm{left}}]])/2$\;
    }{
      $\mathrm{continuity\_penalty} \gets (\mathrm{continuity\_penalty} + \mathbb{1}[d^{\mathrm{target}}[i+1] \neq W[i-idx_{\mathrm{left}}]])/2$\;
    }
  }

  $score \gets \mathrm{mismatch} + \mathrm{distance\_penalty} +  \mathrm{continuity\_penalty}$\;

  \If{$score < score^{'}$}{
    $score^{'} \gets score$, $idx^{'} \gets i$\;
  }
}

\Return $idx^{'}$\;
\end{algorithm}

To mitigate the overhead arising from this mismatch, we introduce a trace mutation algorithm in \algorithmcfname~\ref{alg:tracemutation}. 
Specifically, $d^{target}$ represents the direction of the target morphing trace. $d^{in}$ denotes the direction of the real-time browsing trace. In the trace mutation process, only the directions from the trace are considered. The parameter $\ell$ indicates the starting position of the trace mutation within the target morphing trace. A window size $M$ defines the valid candidate region for mutation around $\ell$ (line 3), restricting the search space to feasible positions. Then, \C iterates over all candidate indices within the window to identify the optimal mutation location (line 9). For each candidate position $i$, \C computes a score to find the best position to execute a trace mutation. 
The smaller score represents the closest position with low overhead and less information leakage.

The computation of the score includes three components. \textbf{First}, the mismatch term quantifies the inconsistency between the mutated target segment and the real-time input trace, reflecting alignment quality (line 12). A higher mismatch indicates greater deviation and potential bandwidth inefficiency due to increased dummy padding. \textbf{Second}, the distance penalty encourages mutations closer to the intended starting position $\ell$ and discourages excessive displacement within the morphing trace (line 13). \textbf{Third}, the continuity penalty is computed to enforce local structural smoothness (lines 15-24). This term aggregates up to two indicator-based comparisons between the mutated symbol at position $i$ and its adjacent neighbors. The continuity penalty discourages abrupt directional flips and promotes locally coherent trace patterns. Finally, \C selects the candidate index with the minimum score as the optimal mutation location. A lower score indicates a better trade-off between trace alignment, spatial consistency, and structural smoothness.

\section{Performance Evaluation}
\label{sec:evaluation}


\subsection{Experiment Setup}
\label{subsec:experiment}

\subsubsection{Dataset}
\label{subsec:dataset}

We conduct experiments on three widely used datasets, DF~\cite{sirinam2018deep} and DS-19~\cite{gong2020zero}, and the GTT23 dataset~\cite{jansen2026measurement}. 

The DF dataset contains both closed-world and open-world data. In the DF closed-world dataset, traffic traces are collected from the homepages of the top 100 Alexa-ranked websites, of which 95 sites remain after preprocessing, with 1,000 instances per site. 
In the open-world dataset, the unmonitored set is constructed from websites ranked within Alexa's top 50,000, excluding those used in the closed-world setting, resulting in 40,716 sites with one instance per site. All traces are collected using the Tor Browser Crawler \cite{juarez2014critical}, with network traffic captured independently via \textit{tcpdump}.

The DS-19 dataset~\cite{gong2020zero} is collected using Tor Browser 8.5a7 with Tor version 0.4.0.1-alpha, automated through command-line execution. The monitored set consists of the homepages of the top 100 Alexa-ranked websites, each visited 100 times. The unmonitored set includes 10,000 websites, with a single instance recorded per site. 

The GTT23~\cite{jansen2026measurement} dataset is the first large-scale genuine website fingerprinting dataset collected from a Tor exit relay, containing approximately ($1.4\times10^{7}$) website visits from over ($1.1\times10^{6}$) domains. We use CELLSHIFT~\cite{jansen2026cellshift} to construct the closed-world dataset with 100 domains, each containing 2,000 instances, and an additional ($1.0\times10^{4}$) domains for the open-world evaluation.

\subsubsection{WF Attacks} We select five SOTA WF attack models as our benchmark. The attacks are listed as follows:

\begin{itemize}
    \item \textbf{DF}~\cite{sirinam2018deep}. Deep Fingerprinting (DF) is a DL–based WF attack that uses CNNs to learn distinctive patterns from traffic trace direction automatically.
    \item \textbf{Var-CNN}~\cite{bhat2018var}. Var-CNN employs a ResNet-based architecture to extract features from packet sequences.
    \item \textbf{RF}~\cite{shen2023subverting}. Robust Fingerprinting (RF) is a CNN-based WF attack that leverages a Traffic Aggregation Matrix (TAM) to capture informative features leaked from Tor traces.
    \item \textbf{NetCLR}~\cite{bahramali2023realistic}. NetCLR employs both semi-supervised and self-supervised learning to augment network traces.
    \item \textbf{Holmes}~\cite{deng2024robust}. Holmes leverages temporal and spatial traffic distributions with contrastive learning to identify websites from partial webpage loading traffic.
\end{itemize}

\subsubsection{WF Defense Baselines}
To make a comprehensive comparison, we select six SOTA WF defense mechanisms as baselines. The defense models are listed as follows:

\begin{itemize}
    \item \textbf{FRONT~\cite{gong2020zero}}. FRONT is a lightweight defense mechanism by introducing trace-to-trace randomness through the randomized insertion of dummy packets.
    \item \textbf{ALERT~\cite{qiao2024trace}}. ALERT generates adversarial perturbations without requiring knowledge of individual traffic traces and produces user-specific universal perturbations.
    \item \textbf{Palette~\cite{shen2024real}}. Palette prevents adversaries from distinguishing among similar websites within a cluster, thereby providing stronger anonymity guarantees.
    \item \textbf{GAPDiS~\cite{xie2025gapdis}}. GAPDiS is an adversarial perturbation-based WF defense method.
    \item \textbf{FRUGAL~\cite{wang2026cease}}. FRUGAL is a traffic obfuscation defense mechanism that minimizes the mutual information between website traffic and webpage labels by strategically injecting dummy packets.
    \item \textbf{Adaptive Tamaraw~\cite{khajavi2026lightening}}. Adaptive Tamaraw integrates regularization-based defenses with dynamic clustering to adjust defense parameters.
\end{itemize}

\subsubsection{Open Science}
To support reproducibility and facilitate further research, we will release the implementation later, including trace processing scripts and evaluation frameworks. We also provide detailed descriptions of dataset construction and experimental settings. The dataset used in this work is publicly available, and we provide detailed preprocessing procedures in the supplementary material. We also report all hyperparameters and implementation details necessary to reproduce the experimental results, as shown in \appendixautorefname~\ref{sec:appendparameter}.

\subsection{Performance and Overhead Comparison}
\label{subsec:performance}

\subsubsection{Closed-world Scenario}

\begin{table*}[t]
\centering
\small
\setlength{\tabcolsep}{1.6pt}
\renewcommand{\arraystretch}{1.2}

\caption{Evaluation of overhead and defense against adversarial training-based attacks across six SOTA WF defenses.}
\label{tab:closedworld}
\scalebox{0.83}{
\begin{tabular}{p{1.6cm}|cc|cccc|cc|cccc}
\toprule
\textbf{\makecell{Dataset}}                 &  \multicolumn{6}{c|}{\cellcolor{olive!15}\textbf{DF~\cite{sirinam2018deep}}}           & \multicolumn{6}{c}{\cellcolor{orange!20}\textbf{DS-19~\cite{gong2020zero}}}        \\
\makecell{\multirow{2}{*}{\textbf{Defenses}}} & \multicolumn{2}{c}{\textbf{ Overhead (\%)}} & \multicolumn{4}{c|}{\textbf{Attack Accuracy (\%)}} & \multicolumn{2}{c}{\textbf{ Overhead (\%)}} & \multicolumn{4}{c}{\textbf{Attack Accuracy (\%)}} \\
\cmidrule(r){2-3} \cmidrule(r){4-7} \cmidrule(r){8-9} \cmidrule(r){10-13}
& Bandwidth & Time & DF  & Var-CNN  & RF & NetCLR & Bandwidth & Time & DF  & Var-CNN  & RF & NetCLR \\

\midrule
\makecell{NoDef} & 0.0 & 0.0 &98.33 & 98.83 & 98.93 & 98.53 & 0.0 & 0.0 & 97.23 & 96.75 & 90.94 & 97.50   \\

\rowcolor{gray!10}
\makecell{FRONT} & \cellcolor{olive!10} 64.46
& \cellcolor{olive!10}0.0
& \cellcolor{olive!10}95.83 {\footnotesize{\color{cyan}{$\downarrow$2.50}}}
& \cellcolor{olive!10}98.15 {\footnotesize{\color{cyan}{$\downarrow$0.68}}} 
& \cellcolor{olive!10}98.20 {\footnotesize{\color{cyan}{$\downarrow$0.73}}}
& \cellcolor{olive!10}95.93  {\footnotesize{\color{cyan}{$\downarrow$2.60}}}
& \cellcolor{orange!10}31.68
& \cellcolor{orange!10}0.0
& \cellcolor{orange!10}96.38 {\footnotesize{\color{cyan}{$\downarrow$0.85}}}
& \cellcolor{orange!10}94.78  {\footnotesize{\color{cyan}{$\downarrow$1.97}}}
& \cellcolor{orange!10}89.17 {\footnotesize{\color{cyan}{$\downarrow$1.77}}}
& \cellcolor{orange!10}96.50 {\footnotesize{\color{cyan}{$\downarrow$1.00}}} \\


\makecell{ALERT} & 25.03
& 201.34
& 69.43 {\footnotesize{\color{cyan}{$\downarrow$28.90}}}
& 75.99 {\footnotesize{\color{cyan}{$\downarrow$22.84}}} 
& 88.03 {\footnotesize{\color{cyan}{$\downarrow$10.90}}}
& 69.57 {\footnotesize{\color{cyan}{$\downarrow$28.96}}}
&  35.44
&  158.23
& 57.16 {\footnotesize{\color{cyan}{$\downarrow$40.07}}}
& 69.59 {\footnotesize{\color{cyan}{$\downarrow$21.16}}}
& 74.46 {\footnotesize{\color{cyan}{$\downarrow$16.48}}}
& 71.57 {\footnotesize{\color{cyan}{$\downarrow$25.96}}} \\


\rowcolor{gray!10}
\makecell{Palette}& \cellcolor{olive!10} 82.16
& \cellcolor{olive!10} 8.18
& \cellcolor{olive!10} 37.25 {\footnotesize{\color{orange}{$\downarrow$61.08}}}
& \cellcolor{olive!10} 45.02 {\footnotesize{\color{orange}{$\downarrow$53.81}}} 
& \cellcolor{olive!10} 51.41 {\footnotesize{\color{cyan}{$\downarrow$47.52}}}
& \cellcolor{olive!10} 31.77 {\footnotesize{\color{orange}{$\downarrow$66.76}}}
& \cellcolor{orange!10} 75.92
& \cellcolor{orange!10} 13.67
& \cellcolor{orange!10} 10.73 {\footnotesize{\color{orange}{$\downarrow$86.50}}}
& \cellcolor{orange!10} 12.38 {\footnotesize{\color{orange}{$\downarrow$84.37}}}
& \cellcolor{orange!10} 32.66 {\footnotesize{\color{orange}{$\downarrow$58.28}}}
& \cellcolor{orange!10} 11.22 {\footnotesize{\color{orange}{$\downarrow$86.31}}}  \\


\makecell{GAPDiS} & 17.96
& 212.41
& 91.87 {\footnotesize{\color{cyan}{$\downarrow$6.46}}}
& 95.17 {\footnotesize{\color{cyan}{$\downarrow$3.66}}}
& 92.38 {\footnotesize{\color{cyan}{$\downarrow$6.55}}}
& 88.23 {\footnotesize{\color{cyan}{$\downarrow$10.30}}}
& 36.77
& 171.36
& 77.04 {\footnotesize{\color{cyan}{$\downarrow$20.19}}}
& 87.16 {\footnotesize{\color{cyan}{$\downarrow$9.59}}}
& 81.02 {\footnotesize{\color{cyan}{$\downarrow$9.92}}}
& 87.36 {\footnotesize{\color{cyan}{$\downarrow$10.17}}} \\

\rowcolor{gray!10}
\makecell{FRUGAL} & \cellcolor{olive!10} 80.00
& \cellcolor{olive!10} 0.0
& \cellcolor{olive!10} 96.59 {\footnotesize{\color{cyan}{$\downarrow$1.74}}} 
& \cellcolor{olive!10} 96.25 {\footnotesize{\color{cyan}{$\downarrow$2.58}}}
& \cellcolor{olive!10} 30.96 {\footnotesize{\color{orange}{$\downarrow$67.97}}}
& \cellcolor{olive!10} 96.78 {\footnotesize{\color{cyan}{$\downarrow$1.75}}}
& \cellcolor{orange!10} 80.00
& \cellcolor{orange!10} 0.0
& \cellcolor{orange!10} 95.42 {\footnotesize{\color{cyan}{$\downarrow$1.81}}}
& \cellcolor{orange!10} 94.63 {\footnotesize{\color{cyan}{$\downarrow$2.12}}}
& \cellcolor{orange!10} 11.04 {\footnotesize{\color{orange}{$\downarrow$79.90}}}
& \cellcolor{orange!10} 95.97 {\footnotesize{\color{cyan}{$\downarrow$1.53}}} \\

\makecell{Adaptive \\ Tamaraw} & 198.23
& 36.62
& 28.47 {\footnotesize{\color{orange}{$\downarrow$67.36}}}
& 33.79 {\footnotesize{\color{orange}{$\downarrow$ 65.04}}} 
& 37.26 {\footnotesize{\color{orange}{$\downarrow$61.67}}}
& 25.27 {\footnotesize{\color{orange}{$\downarrow$73.26}}}
& 140.17
& 18.01
& 13.07 {\footnotesize{\color{orange}{$\downarrow$84.16}}}
& 14.37 {\footnotesize{\color{orange}{$\downarrow$82.38}}}
& 19.36 {\footnotesize{\color{orange}{$\downarrow$71.58}}}
& 12.73 {\footnotesize{\color{orange}{$\downarrow$84.77}}} \\

\rowcolor{gray!10}
\makecell{\textbf{\C}} & \cellcolor{olive!10}130.60
& \cellcolor{olive!10} 14.51
& \cellcolor{olive!10}\textbf{19.03} {\footnotesize{\color{red}{$\downarrow$79.30}}}
& \cellcolor{olive!10}\textbf{26.23} {\footnotesize{\color{red}{$\downarrow$72.60}}}
& \cellcolor{olive!10}\textbf{23.57} {\footnotesize{\color{red}{$\downarrow$75.36}}} 
& \cellcolor{olive!10}\textbf{19.72} {\footnotesize{\color{red}{$\downarrow$76.21}}}
& \cellcolor{orange!10}72.48
& \cellcolor{orange!10} 9.97
& \cellcolor{orange!10}\textbf{7.24}  {\footnotesize{\color{red}{$\downarrow$89.99}}}
& \cellcolor{orange!10}\textbf{6.50} {\footnotesize{\color{red}{$\downarrow$90.25}}}
& \cellcolor{orange!10}\textbf{6.77} {\footnotesize{\color{red}{$\downarrow$84.17}}}
& \cellcolor{orange!10}\textbf{6.40} {\footnotesize{\color{red}{$\downarrow$91.10}}} \\

\bottomrule
\end{tabular}
}
\end{table*}

We first evaluate the effectiveness and robustness of our proposed defense against four SOTA WF attacks under the closed-world setting using four representative attack models on the DF dataset, which contains 95 website classes with 1,000 traces per class, and the DS-19 dataset, which contains 100 website classes with 100 traces per class. As shown in \tableautorefname~\ref{tab:closedworld}, all attacks achieve near-perfect performance without defenses, exceeding 98\% on the DF dataset and remaining above 90\% on the DS-19 dataset. This confirms that the SOTA WF attacks can reliably infer user activity in the absence of effective countermeasures.

From \tableautorefname~\ref{tab:closedworld}, we can see that SOTA defense models such as FRONT, GAPDiS, and FRUGAL show limited robustness under adversarial training-based attacks. They marginally degrade attack performance, with accuracies still exceeding 90\% on the DF dataset and remaining above 80\% on the DS-19 dataset. For example, although FRONT introduces no additional packet delay and incurs only moderate bandwidth overhead, it reduces the DF attack model's accuracy from 98.33\% to merely 95.83\%, indicating minimal defensive effectiveness against adversarial training-based attacks. The weaker performance of FRUGAL compared with its originally reported results is mainly due to our evaluation setting, where 90\% of the defended traces are used for training and 10\% for testing, providing the attacker with substantially more training data for adversarial training.


Palette~\cite{shen2024real} and Adaptive Tamaraw~\cite{khajavi2026lightening} demonstrate stronger resilience by reshaping the trace pattern distribution. Palette reduces attack accuracy by at least 47.52\% on the DF dataset and 58.28\% on the DS-19 dataset, while Adaptive Tamaraw further improves robustness, achieving reductions of at least 61.67\% and 71.58\%, respectively. This highlights the advantage of these two defense models in defending against adversarial training-based attacks. However, Adaptive Tamaraw achieves stronger protection at the expense of considerably higher overhead.

In comparison, our approach \C consistently outperforms all baselines across both datasets on four SOTA attack models. \C keeps attack accuracy below 27\% across all four attacks on the DF dataset. On average, \C reduces the accuracy of SOTA WF attacks by 76.52\% with moderate overhead on the DF dataset. On the DS-19 dataset, \C reduces accuracy as low as 6.40\%, corresponding to over 91.10\% reduction relative to the undefended baseline (NoDef) with lower time overhead. 

Importantly, \C achieves strong defense performance with moderate overhead, incurring 130.60\% bandwidth overhead and only 14.51\% time overhead on the DF dataset.
We found that the DF dataset exhibits substantial variation in trace length across classes, which leads \C to select highly similar candidate traces at the cost of increased padding and higher overhead under imperfect alignment.
In contrast, \C achieves 72.48\% bandwidth with 9.97\% time overhead on the DS-19 dataset, offering a significantly better accuracy–efficiency trade-off than prior defenses. The DS-19 dataset presents more uniform trace lengths, enabling more efficient morphing with even reduced overhead. Overall, these results demonstrate that \C provides substantially stronger and more stable robustness against adversarially trained WF attacks while maintaining practical deployment overhead.

\subsubsection{Open-world Scenario}
\label{subsec:openworld}

\begin{table*}[t]
\centering
\small
\setlength{\tabcolsep}{1.6pt}
\renewcommand{\arraystretch}{1.2}

\caption{Evaluation of adversarial training effectiveness across defenses using DF~\cite{sirinam2018deep} datasets in the open-world scenario.}
\label{tab:openworld}

\scalebox{0.83}{
\begin{tabular}{p{1.9cm}|cccc|cccc|cccc|cccc}
\toprule
\cellcolor{white} & \multicolumn{4}{c|}{\cellcolor{olive!25}\textbf{DF}} & \multicolumn{4}{c|}{\cellcolor{orange!25}\textbf{Var-CNN}} & \multicolumn{4}{c|}{\cellcolor{cyan!25}\textbf{ RF}} & \multicolumn{4}{c}{\cellcolor{pink!25}\textbf{NetCLR}} \\
\cmidrule(r){2-5} \cmidrule(r){6-9} \cmidrule(r){10-13} \cmidrule(r){14-17}
\multirow{-2}{*}{\diagbox[width=1.9cm]{\textbf{Defenses}}{\textbf{Attacks}}}  & \cellcolor{olive!25} TPR (\%) & \cellcolor{olive!25} FPR (\%) & \cellcolor{olive!25} F1  & \cellcolor{olive!25} P (\%)  & \cellcolor{orange!25} TPR (\%)  & \cellcolor{orange!25} FPR (\%) & \cellcolor{orange!25} F1 (\%) & \cellcolor{orange!25} P (\%) & \cellcolor{cyan!25} TPR (\%) & \cellcolor{cyan!25} FPR (\%) & \cellcolor{cyan!25} F1 (\%) & \cellcolor{cyan!25} P (\%)  & \cellcolor{pink!25} TPR (\%)  & \cellcolor{pink!25} FPR (\%) & \cellcolor{pink!25} F1 (\%) & \cellcolor{pink!25} P (\%) \\

\midrule
\makecell{NoDef} & 97.71 
& 3.34
& 98.13
& 98.56
& 98.25
& 1.68
& 98.76
& 98.22
& 93.59 
& 12.37
& 94.11
& 94.64
& 98.24
& 1.62
& 98.77
& 98.31 \\

\rowcolor{gray!25}
\makecell{FRONT} & \cellcolor{olive!15} 94.37
& \cellcolor{olive!15} 8.10
& \cellcolor{olive!15} 95.31
& \cellcolor{olive!15} 96.29
& \cellcolor{orange!15} 97.16
& \cellcolor{orange!15} 2.87
& \cellcolor{orange!15} 97.91
& \cellcolor{orange!15} 98.07
& \cellcolor{cyan!15} 70.20
& \cellcolor{cyan!15} 19.87
& \cellcolor{cyan!15} 78.35
& \cellcolor{cyan!15} 88.75
& \cellcolor{pink!15} 94.34
& \cellcolor{pink!15} 4.83
& \cellcolor{pink!15} 96.01
& \cellcolor{pink!15} 97.73 \\

\makecell{ALERT} & 68.86
& 3.19
& 80.75
& 97.35
& 78.70
& 2.46
& 86.82
& 95.87
& 81.58
& 3.67
& 88.07
& 94.75
& 69.87
& 2.27
& 81.20
& 96.05
\\

\rowcolor{gray!25}
\makecell{Palette}& \cellcolor{olive!15} 30.92
& \cellcolor{olive!15} 42.01
& \cellcolor{olive!15} 41.52
& \cellcolor{olive!15} 62.98
& \cellcolor{orange!15} 36.87
& \cellcolor{orange!15} 38.47
& \cellcolor{orange!15} 48.09
& \cellcolor{orange!15} 69.13
& \cellcolor{cyan!15} 42.33
& \cellcolor{cyan!15} 31.26
& \cellcolor{cyan!15} 54.36
& \cellcolor{cyan!15} 75.94
& \cellcolor{pink!15} 24.02
& \cellcolor{pink!15} 43.30
& \cellcolor{pink!15} 33.90
& \cellcolor{pink!15} 57.76
\\

\makecell{GAPDiS} & 97.54
& 3.77
& 97.96
& 98.39
& 97.16
& 3.04
& 97.91
& 98.17
& 9.58
& 72.67
& 13.62
& 23.31
& 97.38
& 2.72 
& 98.09
& 98.51 \\

\rowcolor{gray!25}
\makecell{FRUGAL} & \cellcolor{olive!15} 96.22
& \cellcolor{olive!15} 6.29
& \cellcolor{olive!15} 96.74
& \cellcolor{olive!15} 97.11
& \cellcolor{orange!15} 95.54
& \cellcolor{orange!15} 4.87
& \cellcolor{orange!15} 96.69
& \cellcolor{orange!15} 97.98
& \cellcolor{cyan!15} 9.04
& \cellcolor{cyan!15} 45.42
& \cellcolor{cyan!15} 14.07
& \cellcolor{cyan!15} 80.82
& \cellcolor{pink!15} 95.70
& \cellcolor{pink!15}  4.25
& \cellcolor{pink!15}  96.90
& \cellcolor{pink!15}  94.90 \\

\makecell{Adaptive \\ Tamaraw} & 11.06
& 5.84
& 18.99
& 66.88
& 16.05
& 2.96
& 27.08
& 78.86
& 19.01
& 1.97
& 31.79
& 86.03
& 13.28
& 1.43
& 23.11
& 88.43 \\

\rowcolor{gray!25}
\makecell{\textbf{\C} } & \cellcolor{olive!15} 13.34
& \cellcolor{olive!15} 7.27
& \cellcolor{olive!15} 22.92
& \cellcolor{olive!15} 54.31
& \cellcolor{orange!15} 8.45
& \cellcolor{orange!15} 11.77
& \cellcolor{orange!15} 14.89
& \cellcolor{orange!15} 31.29
& \cellcolor{cyan!15} 11.04
& \cellcolor{cyan!15} 5.65
& \cellcolor{cyan!15} 19.47
& \cellcolor{cyan!15} 17.91
& \cellcolor{pink!15} 17.72
& \cellcolor{pink!15} 10.40
& \cellcolor{pink!15} 26.16
& \cellcolor{pink!15} 49.79
\\

\bottomrule
\end{tabular}
}
\vspace{-10pt}
\end{table*}

In an open-world scenario, the adversary can only know and monitor a subset of the websites accessed by the user. Traffic traces collected from monitored websites are used to train a classifier that determines whether a given trace corresponds to one of the monitored sites. This setting is inherently more challenging for an adversary to execute a WF attack.
To simulate the open-world scenario, for the DF dataset, we construct the monitored set using all websites in the closed-world dataset, while treating all websites in the open-world dataset as unmonitored. Specifically, we randomly sample 900 traces per monitored website and 40,716 traces from the unmonitored set. Similarly, for the DS-19 dataset, we sample 90 traces per monitored website and 10,000 traces from the unmonitored set. The two datasets are partitioned into training and validation subsets with a 9:1 ratio.

\begin{figure*}[htbp]
    \centering
    \includegraphics[width=0.24\textwidth]{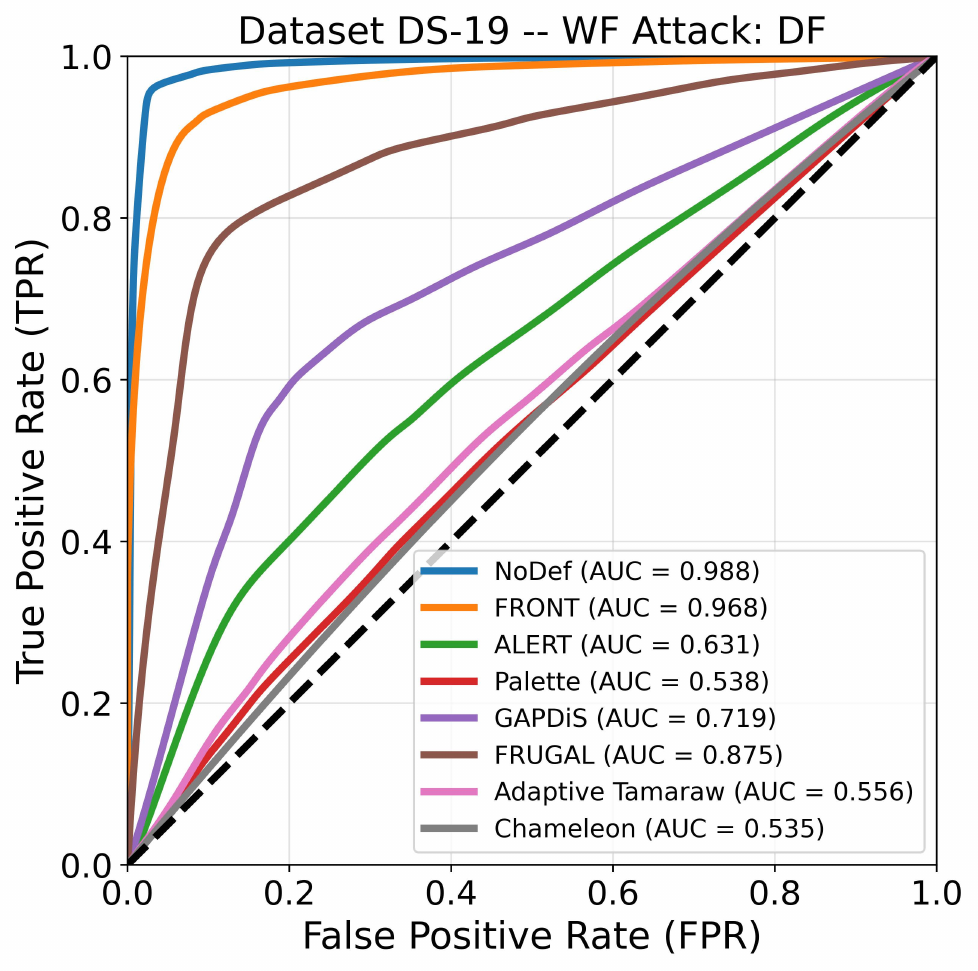}
    \hfill
    \includegraphics[width=0.24\textwidth]{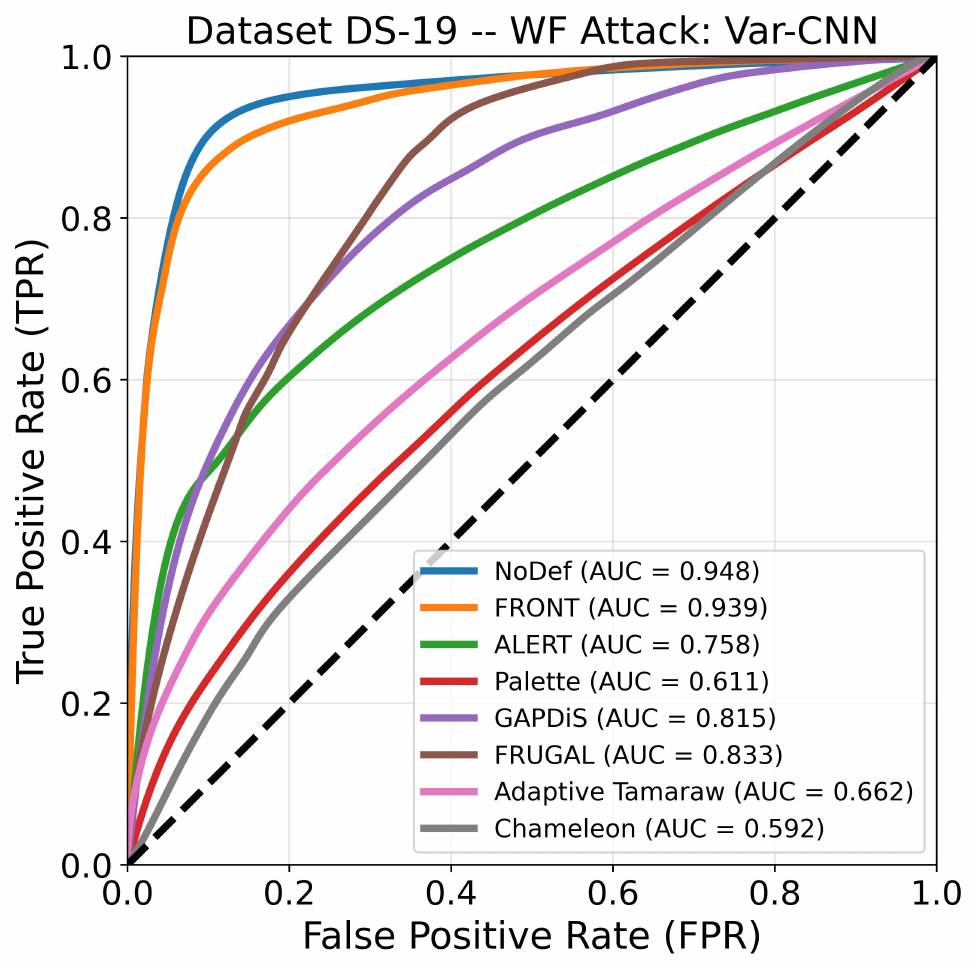}
    \hfill
    \includegraphics[width=0.24\textwidth]{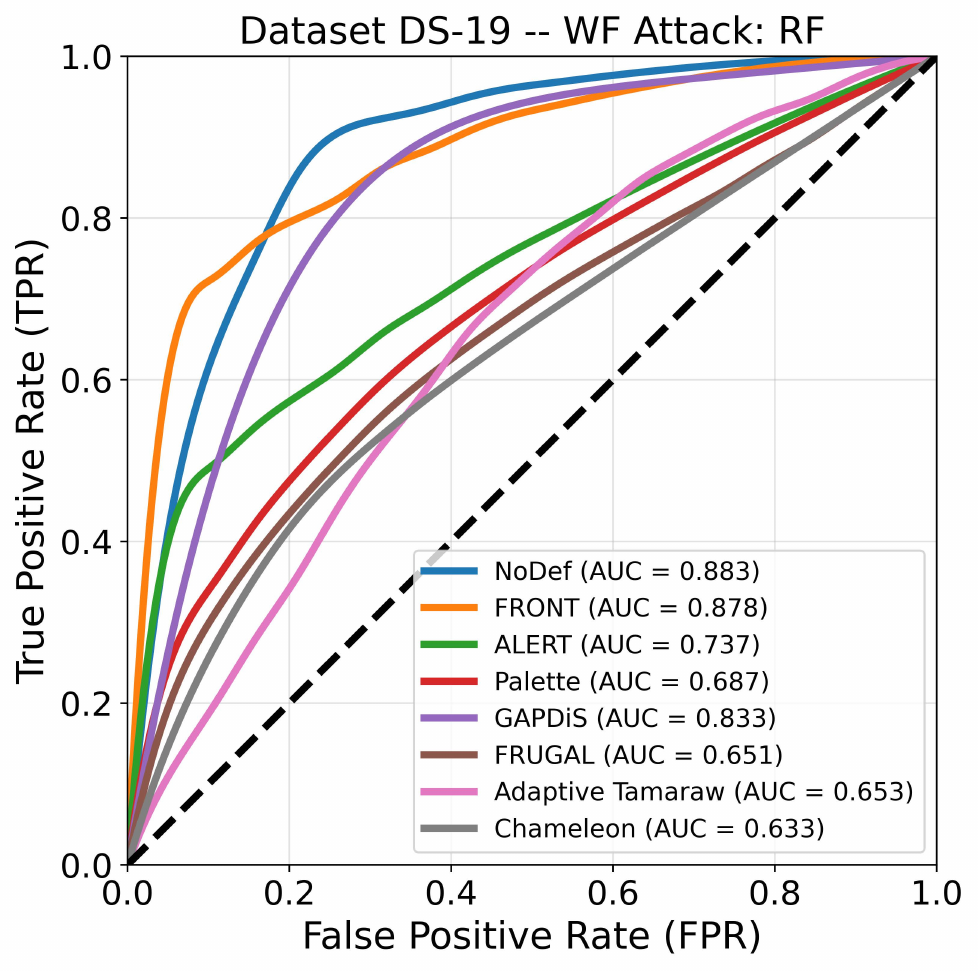}
    \hfill
    \includegraphics[width=0.24\textwidth]{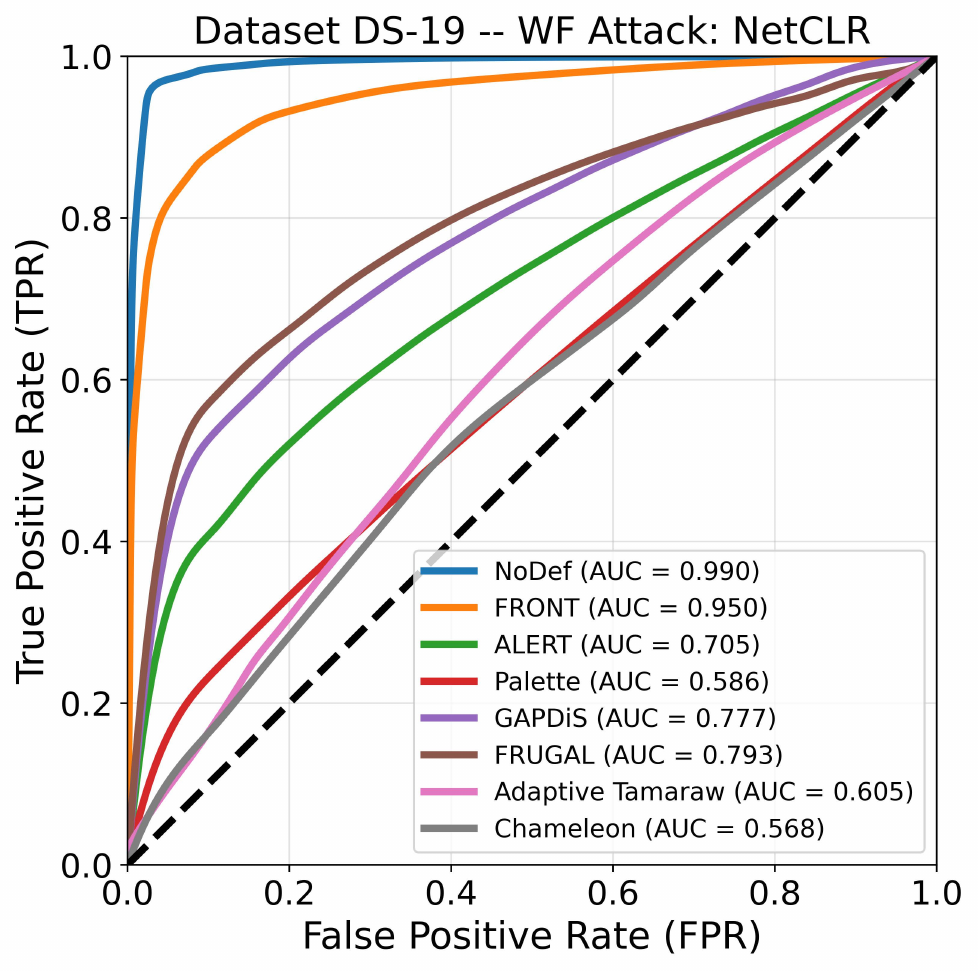}

    \caption{ROC curves and AUC scores of defense methods on the DS-19 dataset under the open-world scenario.}
    \label{fig:openworldds19}
\end{figure*}

Following the evaluation setting in Palette~\cite{shen2024real} and GAPDiS~\cite{xie2025gapdis}, we report performance using True Positive Rate (TPR), False Positive Rate (FPR), and F1-score as shown in \tableautorefname~\ref{tab:openworld}. 
We compare defenses under four SOTA attacks on the DF dataset. Without defenses, DF, Var-CNN, RF, and NetCLR achieve high TPR and F1-scores, demonstrating the strong effectiveness of these four website fingerprinting attacks even in the open-world setting. 
Defenses such as FRONT remain vulnerable to adversarial training-based attacks in the open-world setting, with TPRs exceeding 90\% in most cases. Similarly, GAPDiS exhibits inconsistent behavior, achieving low TPR against RF at 9.58\% but remaining ineffective against other attacks. 
ALERT reduces TPR to a moderate range, 68.86\% to 81.58\%, but still falls short of providing strong protection. Palette significantly lowers TPR in some cases, such as 24.02\% under NetCLR. However, Palette achieves the performance at the cost of extremely high FPR, up to 43.30\%, which severely limits its practicality. 

In contrast, \C maintains relatively moderate FPR values, 5.65\% to 11.77\%, while achieving low TPR, resulting in a more favorable balance as reflected by its F1-scores. FRUGAL achieves a particularly low TPR of 9.04\% against RF. This result is partly attributable to the available FRUGAL defended traces containing only packet-direction information without timing information, whereas RF relies substantially on timing-related features. However, DF, Var-CNN, and NetCLR still achieve TPRs above 90\% against FRUGAL, indicating that its protection remains limited against adversarially trained attacks. Adaptive Tamaraw achieves low TPRs of 11.06\% under DF and 13.28\% under NetCLR. By contrast, although Adaptive Tamaraw attains slightly lower FPR in some cases, 1.43\% under NetCLR, \C achieves more stable performance across different attacks.
Overall, our design, \C, provides competitive performance while maintaining a more balanced error profile.

In addition, we show the ROC curves of four SOTA WF attacks under different defense mechanisms in \figureautorefname~\ref{fig:openworldds19}. Similar to the observations on the DF dataset, the WF attack models, such as DF, Var-CNN, and NetCLR, achieve high AUC values in the absence of defenses, indicating their strong discriminative capability in the open-world setting. Although the RF attack yields a comparatively lower AUC of 88.3\%, it still demonstrates effectiveness and remains a competitive WF attack.
From \figureautorefname~\ref{fig:openworldds19}, we can see Palette and Adaptive Tamaraw exhibit relatively strong defense performance. 
However, both methods still leave room for improvement. Our designed method, \C, further reduces the attack effectiveness, achieving an AUC as low as 53.5\% against the DF attack. This result highlights the advantage of \C in significantly degrading the adversary's WF attack capability.

\subsection{Evaluation of Defense against DAAE-based Attacks}
\label{subsec:evadaae}

\begin{table*}[t]
\centering
\small
\setlength{\tabcolsep}{5pt}
\renewcommand{\arraystretch}{1.25}

\caption{A comprehensive comparison of AT-resistant defenses against DAAE-based attacks on the GTT23 dataset.}
\label{tab:daaeeva}
\begin{tabular}{c|ccc|ccc|ccc|ccc}
\toprule

\multirow{2}{*}{\diagbox[width=2.2cm]{\textbf{Defenses}}{\textbf{Attacks}}} & \multicolumn{3}{c|}{\textbf{DF}} & \multicolumn{3}{c|}{\textbf{Var-CNN}} & \multicolumn{3}{c|}{\textbf{RF}} & \multicolumn{3}{c}{\textbf{NetCLR}}  \\
\cmidrule(r){2-4} \cmidrule(r){5-7} \cmidrule(r){8-10} \cmidrule(r){11-13}
& P (\%) & R (\%) & F1 (\%)  & P (\%) & R (\%) & F1 (\%)  & P (\%) & R (\%) & F1 (\%)  & P (\%) & R (\%) & F1 (\%)  \\

\midrule
\rowcolor{gray!25}
NoDef & 93.82
& 93.20
& 93.13
& 93.21
& 91.80
& 91.89
& 94.86
& 92.20
& 94.01
& 93.76
& 93.00
& 93.03 \\


$Palette_{def}$ & 10.57
& 8.30
& 7.76
& 27.30
& 21.00
& 19.80
& 31.60
& 24.50
& 23.59
& 10.22
& 8.90
& 7.92 \\

\rowcolor{gray!25}
$Palette_{rec}$ & 86.90
& 83.00
& 83.67
& 86.38
& 82.30
& 82.72
& 80.64
& 67.90
& 70.30
& 81.99
& 80.50
& 80.55\\

$\makecell{Adaptive \\ Tamaraw_{def}}$& 17.47
& 13.79
& 13.58
& 14.05
& 13.37
& 13.19
& 34.97
& 30.33
& 28.35
& 14.76
& 13.41
&  13.16 \\

\rowcolor{gray!25}
$\makecell{Adaptive \\ Tamaraw_{rec}}$ & 87.38
& 86.78
& 86.44
& 88.74
& 87.98
& 87.85
& 89.82
& 88.65
& 88.22
& 85.79
& 84.96
& 84.11 \\

$\C_{def}$ & 8.92
& 7.67
& 6.42
& 12.26
& 12.44
& 9.37
& 20.33
& 19.78
& 16.23
& 7.78
& 6.34
& 5.97 \\

\rowcolor{gray!25}
\textbf{$\C_{rec}$} & 26.13
& 26.41
& 23.86
& 27.92
& 27.91
& 22.45
& 40.68
& 34.97
& 35.19
& 25.13
& 25.28
& 24.02 \\

\bottomrule
\end{tabular}
\vspace{-10pt}
\end{table*}











DAAE-based attack is designed for scenarios in which adversarial training-based WF cannot attack a defense. To evaluate the resistance of \C to DAAE, we compare it with two SOTA defenses, Palette and Adaptive Tamaraw, which demonstrate strong resistance to adversarial training on the recent GTT23~\cite{jansen2026measurement} dataset. We first use DAAE to reconstruct defended traces and then perform WF attacks using DF, Var-CNN, RF, and NetCLR implemented in the WFLib project~\cite{wflib} to evaluate precision (P), recall (R), and F1-score (F1). Table~\ref{tab:daaeeva} presents the performance of DAAE-based attacks against defenses that are resistant to adversarial training on the GTT23 dataset. \C provides substantially stronger protection than the existing defenses, particularly when the attacker is trained on defended traffic. For the defense-trained setting $\C_{def}$, \C achieves F1-scores of 6.42\%, 9.37\%, 16.23\%, and 5.97\% against DF, Var-CNN, RF, and NetCLR, respectively. These results are consistently lower than those of $Palette_{def}$ and $Adaptive~Tamaraw_{def}$, demonstrating that \C is robust and efficient in defending against adversarial training-based attacks.

After performing the DAAE-based WF attacks, the results further highlight the robustness of \C. Under $\C_{rec}$, the F1-score of the DAAE-based DF attack only increases from 6.42\% to 23.86\%, while the same attack increases from 7.76\% to 83.67\% under the $Palette_{rec}$ defense.
In particular, $\C_{rec}$ limits the strongest attack, RF, to an F1-score of only 35.19\%, while the corresponding scores for $Palette_{rec}$ and $Adaptive~Tamaraw_{rec}$ reach 70.30\% and 88.22\%, respectively. These results indicate that \C remains robust and effective even when the attacker can reconstruct or retrain on defended traffic, confirming that its randomized morphing strategy provides stronger resistance to defense-aware attacks.

\begin{table*}[t]
\centering
\small
\setlength{\tabcolsep}{3pt}
\renewcommand{\arraystretch}{1.2}

\caption{The early-stage traffic evaluation through Holmes WF attack~\cite{deng2024robust} under the GTT23~\cite{jansen2026measurement} dataset.}
\label{tab:earlystage}

\begin{tabular}{c|ccc|ccc|ccc|ccc|ccc}
\toprule

\makecell{\multirow{2}{*}{\textbf{Defenses}}} & \multicolumn{3}{c|}{\textbf{20\% loaded}} & \multicolumn{3}{c|}{\textbf{30\% loaded}} & \multicolumn{3}{c|}{\textbf{40\% loaded}} & \multicolumn{3}{c|}{\textbf{50\% loaded}} &  \multicolumn{3}{c}{\textbf{60\% loaded}} \\
\cmidrule(r){2-4} \cmidrule(r){5-7} \cmidrule(r){8-10} \cmidrule(r){11-13} \cmidrule(r){14-16}
& P (\%) & R (\%) & F1 (\%)  & P (\%) & R (\%) & F1 (\%) & P (\%) & R (\%) & F1 (\%)  & P (\%) & R (\%) & F1 (\%) & P (\%) & R (\%) & F1 (\%) \\

\midrule

NoDef & 54.45
& 26.85
& 31.07
& 61.12
& 39.91
& 45.24
& 64.03
& 44.09
& 49.34
& 65.50
& 59.44
& 60.44
& 66.03
& 54.43
& 57.61 \\

\rowcolor{gray!25}
Palette & 13.87
& 7.94
& 7.83
& 15.29
& 9.82
& 9.70
& 17.85
& 11.42
& 11.34
& 18.53
& 13.56
& 13.53
& 19.25
& 16.68
& 16.43\\

\makecell{Adaptive \\ Tamaraw} & 15.38
& 8.21
& 7.72
& 16.57
& 11.75
& 11.11
& 17.42
& 15.73
& 14.19
& 19.33
& 17.19
& 16.76
& 21.04
& 18.39
& 17.49 \\

\rowcolor{gray!25}
\textbf{\C} & 18.95
& 7.64
& 8.06
& 17.75
& 9.48
& 10.02
& 18.26
& 10.75
& 11.25
& 18.83
& 12.72
& 13.01
& 18.88
& 15.06
& 14.67\\

\bottomrule
\end{tabular}
\end{table*}

\begin{table*}[t]
\centering
\small
\setlength{\tabcolsep}{5pt}
\renewcommand{\arraystretch}{1.2}

\caption{Evaluation of performance and overhead trade-offs for \C across full trace morphing, trade-off trace morphing, and trace mutation on the DF~\cite{sirinam2018deep} and DS-19~\cite{gong2020zero} datasets in the closed-world scenario.}
\label{tab:tradeoffcw}
\begin{tabular}{p{2.0cm}|cc|cccc|cc|cccc}
\toprule
\textbf{\cellcolor{gray!15} \makecell{Dataset}}                 &  \multicolumn{6}{c|}{\cellcolor{olive!15}\textbf{DF~\cite{sirinam2018deep}}}           & \multicolumn{6}{c}{\cellcolor{orange!20}\textbf{DS-19~\cite{gong2020zero}}}        \\
\makecell{\multirow{2}{*}{\textbf{Defenses}}} & \multicolumn{2}{c}{\textbf{ Overhead (\%)}} & \multicolumn{4}{c|}{\textbf{Attack Accuracy (\%)}} & \multicolumn{2}{c}{\textbf{ Overhead (\%)}} & \multicolumn{4}{c}{\textbf{Attack Accuracy (\%)}} \\
\cmidrule(r){2-3} \cmidrule(r){4-7} \cmidrule(r){8-9} \cmidrule(r){10-13}
& Bandwidth & Time & DF  & Var-CNN  & RF & NetCLR & Bandwidth & Time & DF  & Var-CNN  & RF & NetCLR \\

\midrule

\rowcolor{gray!10}
\makecell{\C \\ (Lightweight)} & \cellcolor{olive!10}130.60
& \cellcolor{olive!10}14.51
& \cellcolor{olive!10}19.03
& \cellcolor{olive!10}26.23
& \cellcolor{olive!10}23.57
& \cellcolor{olive!10}19.72
& \cellcolor{orange!10}72.48
& \cellcolor{orange!10}9.97
& \cellcolor{orange!10}7.24
& \cellcolor{orange!10}6.50
& \cellcolor{orange!10}6.77
& \cellcolor{orange!10}6.40 \\

\makecell{\C \\ (Medium)} & 172.02
& 17.26
& 16.68
& 18.31 
& 17.69
& 14.61
& 78.82
& 11.70
& 6.44
& 5.75
& 6.32
& 5.95 \\

\rowcolor{gray!10}
\makecell{\C \\ (Heavy)} & \cellcolor{olive!10}260.46
& \cellcolor{olive!10}48.50
& \cellcolor{olive!10}10.49
& \cellcolor{olive!10}9.84
& \cellcolor{olive!10}12.94
& \cellcolor{olive!10}9.73
& \cellcolor{orange!10}93.42
& \cellcolor{orange!10}14.50
& \cellcolor{orange!10}4.17
& \cellcolor{orange!10}2.89
& \cellcolor{orange!10}4.07
& \cellcolor{orange!10}3.18 \\

\bottomrule
\end{tabular}
\end{table*}

\subsection{Early-stage Measurement}
\label{subsec:earlystage}

In this section, we evaluate whether early-stage information used in \C can leak sufficient trace information for website identification. To evaluate this potential leakage, we employ Holmes~\cite{deng2024robust}, an early-stage WF attack that exploits the temporal and spatial characteristics of partially observed traffic under the GTT23 dataset. Following the evaluation methodology of Holmes, we evaluate the first 20\%, 30\%, 40\%, 50\%, and 60\% of the loaded packets and compare \C with Palette and Adaptive Tamaraw. For all defenses, we set the trace length to 5,000 packets. As shown in \tableautorefname~\ref{tab:earlystage}, NoDef exhibits increasing attack performance as more traffic is observed, with the F1-score rising from 31.07\% at 20\% to 57.61\% at 60\% loaded traffic. In contrast, all three defenses substantially suppress Holmes, demonstrating that traffic regularization can effectively reduce the information available during the early stages of a connection.

\C maintains consistently low attack performance across all observation levels, with F1-scores of 8.06\%, 10.02\%, 11.25\%, 13.01\%, and 14.67\% from 20\% to 60\% loaded traffic, respectively. Even at 60\% loaded traffic, Holmes achieves only 14.67\% F1 against \C, representing a 74.5\% reduction compared with NoDef. Moreover, \C achieves lower recall and F1-scores than both Palette and Adaptive Tamaraw at 50\% and 60\% loaded traffic, indicating that its prefix-based trace selection does not expose additional discriminative information as more of the trace becomes available. These results suggest that the partial traffic prefix used for synchronization does not compromise the protection of \C, and that its many-to-many morphing strategy remains effective against early-stage WF attacks.

\subsection{Ablation Study}
\label{subsec:tuning}

\subsubsection{Traffic Morphing Strategy}
We perform a trade-off evaluation between three strategies: full trace morphing (Heavy), morphing trace cut-off (Medium), and trace mutation (Lightweight), as discussed in \sectionautorefname~\ref{subsec:mutation}.
\tableautorefname~\ref{tab:tradeoffcw} presents a trade-off experiment study of \C under three configurations, lightweight, medium, and heavy, capturing different trade-offs between defense strength and system overhead. Under the DF dataset, the trace mutation version (Lightweight) of \C already achieves substantial protection, reducing attack accuracy to 19.03\%, 26.23\%, 23.57\%, and 19.72\% against DF, Var-CNN, RF, and NetCLR, respectively. Meanwhile, the lightweight version of \C achieves moderate overhead on 130.60\% bandwidth and 14.51\% time. The medium version of \C further improves robustness, lowering attack accuracy across all models, such as 16.68\% for DF and 14.61\% for NetCLR, with a moderate increase in overhead. The heavy version of \C provides the strongest defense, reducing attack accuracy to as low as 9.73\% to 12.94\%, but has a significantly higher overhead on 260.46\% bandwidth and 48.50\% latency.

These results demonstrate that \C offers a flexible design space: practitioners can select an appropriate configuration depending on their desired balance between security and overhead. In particular, the lightweight version of \C already provides strong protection with reasonable cost, while the heavy variant achieves near-minimal attack accuracy for high-security scenarios.

\begin{figure}[t]
    \centering
    \begin{minipage}{0.48\linewidth}
        \centering
        \includegraphics[width=\linewidth]{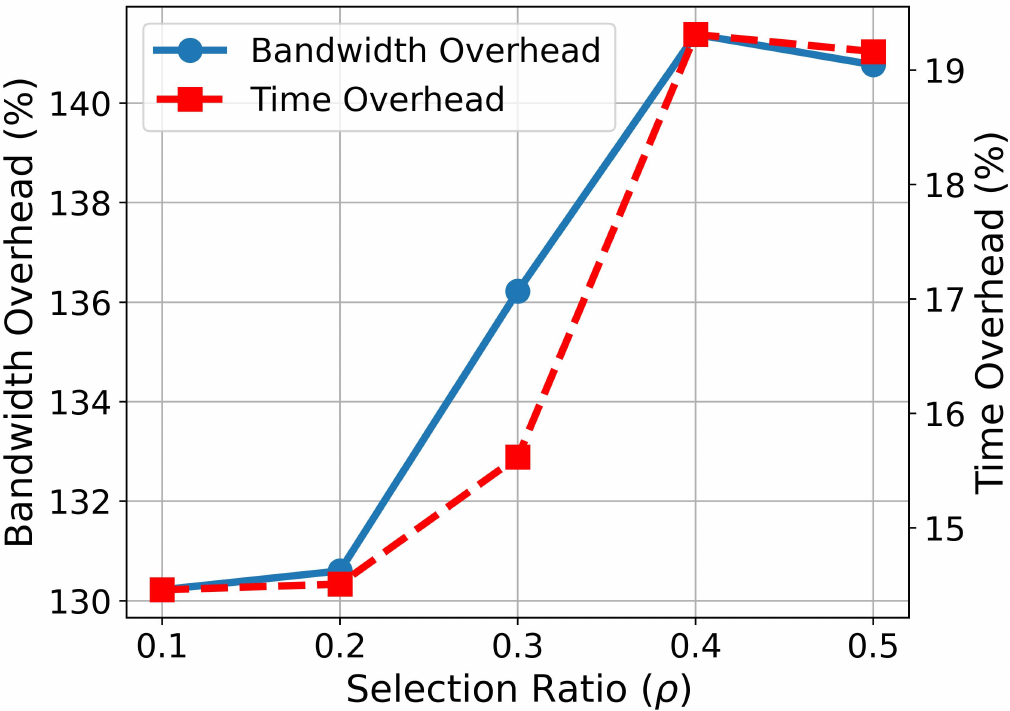}

        \small (a) Overhead on the DF.
    \end{minipage}
    \hfill
    \begin{minipage}{0.48\linewidth}
        \centering
        \includegraphics[width=\linewidth]{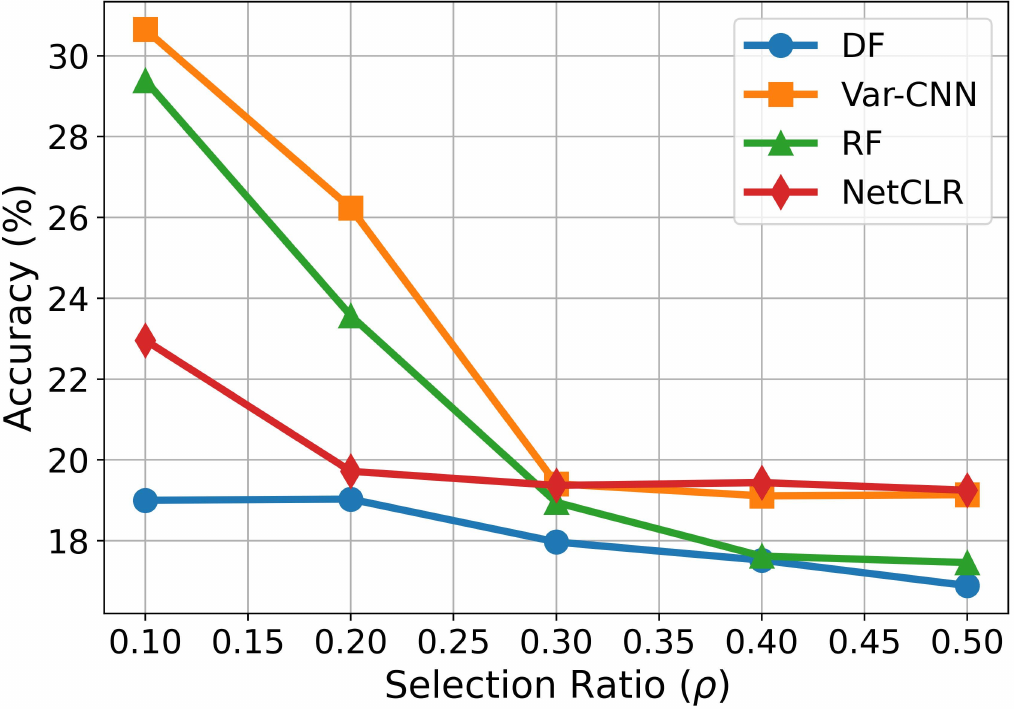}

        \small (b) Accuracy on the DF.
    \end{minipage}

    \caption{The overhead and performance of \C with different values of selection ratio $\rho$ on the DF dataset.}
    \label{fig:ablationdf}
\end{figure}

        
        


\subsubsection{Key Hyperparameter Evaluation}
\label{subsec:ablation}

We evaluate the key hyperparameter $\rho$, which controls the top $m$ website traces for the candidate pool.
The selection ratio $\rho$ governs the granularity of intra-class diversity and inter-class separability by controlling the number of website traces in the candidate pool.
We investigate the impact of the key parameter $\rho$ on \C under four SOTA WF attacks in the closed-world scenario. Intuitively, varying $\rho$ directly affects both the defense performance and the associated overhead.

\figureautorefname~\ref{fig:ablationdf} illustrates the overhead and performance of \C across different values of $\rho$ on the DF dataset. As $\rho$ increases, both bandwidth and time overheads grow, while the attack accuracy of SOTA WF models decreases. This trend arises because a larger $\rho$ enhances intra-class diversity and reduces inter-class separability, thereby making classification more difficult for the attacker. We observe consistent behavior on the DS-19 dataset: increasing $\rho$ leads to higher overhead but lower attack accuracy. To balance efficiency and robustness, we select $\rho=0.2$ in our experiments.

\begin{table}[t]
\centering
\small
\setlength{\tabcolsep}{6pt}
\renewcommand{\arraystretch}{1}

\caption{Robustness evaluation of \C against randomization-aware training on the GTT23 dataset.}
\label{tab:random}
\begin{tabular}{ccccccc}
\toprule

\multirow{2}{*}{\textbf{Number}} & \multicolumn{4}{c}{\textbf{Attack Accuracy (\%)}} \\
\cmidrule(r){2-5}
&  DF  &  Var-CNN  &   RF &  NetCLR  \\

\midrule
\rowcolor{gray!25}
\textbf{$K=1$}
& 7.10
& 6.72
& 6.49
& 5.84 \\

\textbf{$K=5$}
& 13.57
& 15.04
& 20.33
& 14.05 \\

\rowcolor{gray!25}
\textbf{$K=10$}
& 15.93
& 17.63
& 26.42
& 15.99 \\

\textbf{$K=15$}
& 17.02
& 19.12
& 29.29
& 16.91 \\

\rowcolor{gray!25}
\textbf{$K=20$}
& 17.50
& 18.82
& 31.06
&  17.04 \\

\bottomrule
\end{tabular}
\vspace{-10pt}
\end{table}

\subsubsection{Robustness to Randomization-Aware Training}
\label{subsec:robustness}
Since \C relies on randomized traffic transformations, evaluating it using only one defended trace per training instance may overestimate its effectiveness~\cite{mathews2023sok}. To provide a stronger evaluation, we consider an adaptive attacker that augments its training set with multiple independently randomized defended traces generated from each original training trace. Specifically, for each training trace, we generate 20 defended variants using independent random seeds, where $K \in \{1,5,10,15,20\}$. We then train the attacker on the resulting augmented dataset while keeping the test set fixed and independently randomized. Table~\ref{tab:random} evaluates \C's robustness against randomization-aware training on the GTT23 dataset. We augment each original training trace with multiple independently randomized defended traces, increasing the number of replicates from $K=1$ to $K=20$. As expected, exposing the attacker to more randomized realizations improves attack accuracy across all four WF models, indicating that the attacker can gradually learn features that are invariant to \C's randomization. Nevertheless, the attack accuracy remains relatively low even with 20 replicates. Specifically, DF, Var-CNN, RF, and NetCLR achieve only 17.50\%, 18.82\%, 31.06\%, and 17.04\% accuracy at $K=20$, respectively, compared with 7.10\%, 6.72\%, 6.49\%, and 5.84\% at $K=1$. Although RF benefits the most from increased training diversity, its accuracy remains substantially below the closed-world baseline of 100 classes, while the other three attacks remain below 20\%. These results demonstrate that \C's robustness and effectiveness are not solely attributable to the limited number of randomized realizations observed during training and that its protection remains effective against attackers explicitly trained on multiple independently randomized traces.

\subsubsection{Evaluation of Injectivity Boundaries and Candidate Trace}
\label{subsec:bound}

In this work, we introduce a defense model based on a many-to-many mapping mechanism and establish a corresponding security bound for \C in \appendixautorefname~\ref{subsec:boundproof}. Evaluating this upper bound is essential for understanding how the size of input website traces impacts adversarial robustness under a fixed-size candidate pool. 

We empirically study the \emph{injectivity boundary} of \C using a fixed candidate pool derived from the DS-19 dataset with a selection ratio $\rho = 0.2$ (833 website traces). As input, we use the DF dataset (95$\times$1000 website traces) to evaluate this boundary. The ratio between the input website traces and the website candidates is 95,000:833. We define $I$ as the number of input traces per class from the DF dataset, which models varying levels of trace availability for morphing. Increasing $I$ effectively enlarges the feasible injection space of the candidate pool, thereby enabling a finer-grained analysis of the injectivity boundary.

\begin{figure}[t]
  \centering
  \includegraphics[scale=0.35]{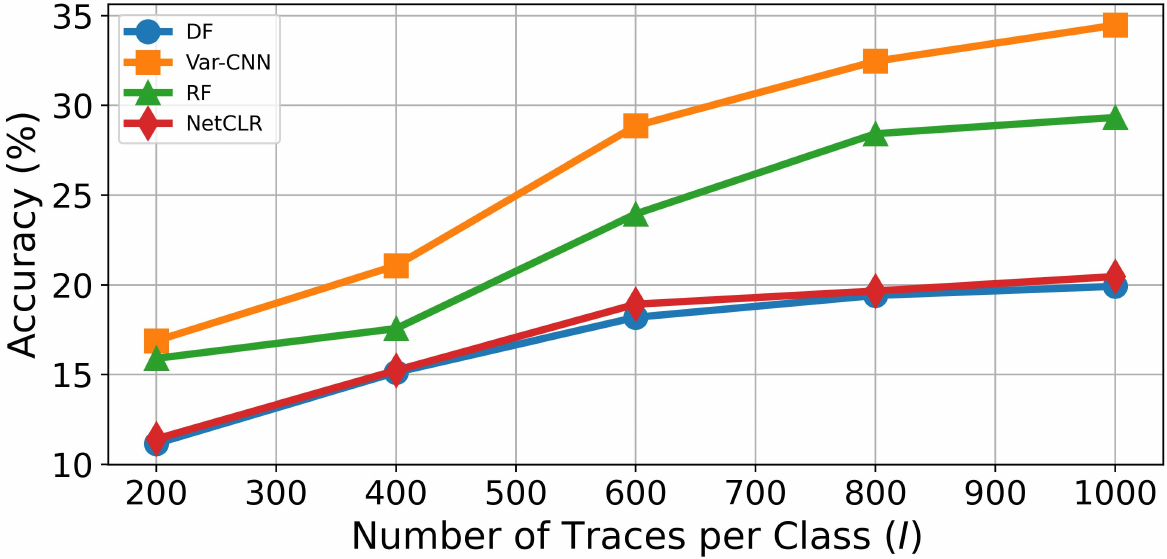}
  \caption{Visualization of Injectivity Boundaries Evaluation.}
  \label{fig:boundary}
  \vspace{-13pt}
\end{figure}

As shown in \figureautorefname~\ref{fig:boundary}, the sensitivity to the injectivity boundary varies across attack models, with Var-CNN exhibiting the most pronounced increase in accuracy, followed by RF and NetCLR, while DF remains comparatively more stable. Even under the input traces $I=1000$, the worst-case WF attack accuracy remains below 35\% across all models, indicating that the defense remains within a reasonable range for defending against adversarial training-based attacks. Furthermore, these results demonstrate that the effectiveness of our defense is not constrained by the candidate-pool dataset. Instead, the candidate pool can be constructed from arbitrary real-world website traces, enabling the defense to generalize beyond a predefined set of websites.

\subsection{Real-world Performance}
\label{subsec:realworld}

To estimate the overhead and performance of \C in the real-world, we implement \C within WFDefProxy~\cite{gong2023wfdefproxy} as a pluggable transport based on the Tor Pluggable Transport (PT) bridge~\cite{torptspec2022}. Our deployment uses two workstations: one hosts the bridge (PT server), while the other runs five Docker-based clients that visit websites in parallel. Each client connects to the PT server under a bandwidth cap of 120 Mbps. The bridge is provisioned with 2 CPU cores and 2 GB RAM, while the client machine has 8 CPU cores and 16 GB RAM; both run Ubuntu 24.04.4 LTS. We automate browsing using TBSelenium~\cite{tor-browser-selenium} with Tor Browser version 15.0. To ensure full page loads, each visit is limited to 80 seconds, followed by an additional 5-second dwell time to avoid overlap with subsequent traffic. After each visit, the browser is restarted, and a \textit{Signal.NEWNYM} command is issued to enforce a fresh circuit.

\begin{table}[t]
\centering
\small
\setlength{\tabcolsep}{3pt}
\renewcommand{\arraystretch}{1}

\caption{Evaluation of real-world performance.}
\label{tab:realworld}
\scalebox{0.92}{
\begin{tabular}{c|cc|cccc}
\toprule

\multirow{2}{*}{Defenses} & \multicolumn{2}{c|}{\textbf{ Overhead (\%)}} & \multicolumn{4}{c}{\textbf{Attack Accuracy (\%)}} \\
\cmidrule(r){2-3} \cmidrule(r){4-7}
& Bandwidth & Time &  DF  &  Var-CNN  &   RF &  NetCLR  \\

\midrule
\rowcolor{gray!25}
NoDef & 0.0
& 0.0
& 90.35
& 88.20
& 91.52
& 89.58 \\

FRONT & 97.04
& 0.0
& 56.16
& 61.50
& 78.64
& 58.39 \\

\rowcolor{gray!25}
RegulaTor & 75.08
& 125.53
& 18.74
& 22.72
& 31.06
& 13.53 \\

Palette & 84.40
& 28.91
& 10.30
& 12.95
& 42.25
& 9.31 \\

\rowcolor{gray!25}
\textbf{\C} & 91.47
& \textbf{16.25}
& \textbf{7.17}
& \textbf{9.33}
& \textbf{17.45}
& \textbf{6.51}  \\

\bottomrule
\end{tabular}
}
\vspace{-10pt}
\end{table}

Real-world evaluation poses greater challenges for WF defenses due to network variability and practical deployment constraints. Following prior work~\cite{shen2024real, gong2023wfdefproxy}, we collect a real-world defended dataset by capturing traffic between the PT client and PT server. Similar to prior studies~\cite{shen2024real}, we select websites from the Tranco~\cite{LePochat2019} list snapshot dated July 1, 2025. From the top 1,000 entries, we remove inaccessible and duplicate URLs that resolve to identical pages, and retain the first 100 sites as the monitored set. Multiple clients visit these sites in randomized order. We aggregate all traces and filter outliers following a prior approach~\cite{panchenko2016website}, ensuring that each monitored website is represented by at least 100 instances.

\tableautorefname~\ref{tab:realworld} shows a comprehensive comparison of the real-world performance and overhead between defense mechanisms such as FRONT, RegulaTor, Palette, and \C. We observe that \C achieves a suitable trade-off between effectiveness and efficiency, reducing attack accuracy to 7.17\% for DF, 9.33\% for Var-CNN, 17.45\% for RF, and 6.51\% for NetCLR, while incurring a moderate bandwidth overhead of 91.47\% and a relatively low time overhead of 16.25\%. In contrast, FRONT provides partial robustness in the real-world, substantially reducing the accuracy of attacks such as DF, Var-CNN, and NetCLR but remaining ineffective against RF, which still achieves 78.64\% accuracy. RegulaTor exhibits a clear discrepancy between effectiveness and efficiency, introducing a large time overhead of 125.53\% while still allowing non-trivial attack success. In particular, the RF attack still reaches 31.06\% accuracy under RegulaTor defense. Similarly, Palette achieves strong reductions for DF, Var-CNN, and NetCLR but remains vulnerable to RF with 42.25\% accuracy and incurs higher time overhead than \C, reflecting a less favorable efficiency–effectiveness balance. These results demonstrate that \C achieves the lowest attack accuracy across all four evaluated attacks while maintaining competitive overhead.


\section{Related Work}
WF defenses aim to obfuscate or alter traffic patterns by inserting dummy packets or increasing packet timing delays that reduce WF performance. However, the existing WF defense methods inevitably introduce new artifacts, which can be exploited by WF attack models. For example, FRONT~\cite{gong2020zero} obfuscates the burst patterns by adding random dummy packets. However, FRONT continues to leak trace information, which remains vulnerable to DL-based WF attacks. TrafficSliver~\cite{de2020trafficsliver} mitigates the observable traffic patterns by splitting traffic across multiple entry relays. However, the DL–based WF attacks can still exploit timing, burst structure, and directional patterns left by TrafficSliver. 
As a result, these defenses exhibit limited robustness against adaptive adversaries, particularly those employing adversarial retraining.

Another main defense design is regularization-based, which reshapes website traffic traces into more uniform patterns, thereby providing rigorous bounds on an adversary's ability to classify webpages. Walkie-Talkie~\cite{wang2017walkie} maps the original website trace to another trace, effectively creating collisions in which at least two traces share identical patterns.
However, Walkie-Talkie does not obfuscate fine-grained timing information, making this defense mechanism vulnerable to advanced DL–based website fingerprinting attacks~\cite{bhat2018var, rahman2019tik, bahramali2023realistic}. 
Additionally, Walkie-Talkie requires the browser to operate in half-duplex mode, which means the client must wait for all responses before issuing new requests. 
Similarly, Palette~\cite{shen2024real} incurs significant buffering overhead to support large-scale clustering and traffic shaping. Adaptive Tamaraw~\cite{khajavi2026lightening} creates an adaptive WF defense framework by integrating Tamaraw~\cite{cai2014systematic} defense with dynamic clustering to adjust defense parameters in real time. However, Adaptive Tamaraw can also incur substantial bandwidth and latency overhead as it dynamically adapts its defense parameters~\cite{khajavi2026lightening}.

\section{Conclusion}

In this work, we show that existing WF defenses can remain vulnerable to defense-aware attacks, even when they are robust against adversarial training.
To address this problem, we present \C, a deployable WF defense based on many-to-many randomized traffic morphing. \C combines diverse shared morphing targets, radix-trie-based synchronization, and trace mutation to disrupt learnable defense patterns while controlling overhead. Across closed-world, open-world, DAAE-based, randomization-aware, early-stage, and real-world evaluations, Chameleon provides stronger protection than existing defenses. These results demonstrate that eliminating defense-induced artifacts is a promising direction for building robust website fingerprinting defense.

\section*{Acknowledgement}
\label{sec:ack}

This work was conducted while Dr. Yao Liu was at the University of South Florida (USF).

\newpage

\section*{Ethical Considerations}
During the study, we followed the guidelines of the Tor Research Safety Board~\cite{torresearchsafetyboard}. All traffic was collected from browsing sessions initiated by our own clients, without involving real Tor user traffic or disrupting the Tor network. While \C is designed to protect web-browsing privacy, we acknowledge that its traffic obfuscation techniques could potentially be misused to evade legitimate network monitoring, such as intrusion detection and abuse mitigation.

\bibliographystyle{IEEEtran}
\bibliography{ref}

\appendix

\section{Hyperparameters \& Validation Experiments}
\label{sec:appendhyper}

\subsection{Hyperparameter Setting}
\label{sec:appendparameter}

To evaluate each WF defense model under an optimal trade-off between performance and overhead, we employ a carefully selected set of hyperparameters based on the configurations reported in the original studies. \tablename~\ref{tab:hyperparameters} summarizes the hyperparameter settings used for each model in this work. 

In FRONT~\cite{gong2020zero}, $N_{c}$ and $N_{s}$ are key parameters that determine the data overhead, representing the padding budgets at the client and proxy, respectively. The parameters $W_{min}$ and $W_{max}$ define the padding window (Minimum and maximum padding time), which controls the temporal range within which dummy packets are most likely to be injected into the original traffic trace. In Palette~\cite{shen2024real}, $u\_upload$ and $u\_download$ are introduced to tune the upper bounds, enabling a principled trade-off between time overhead and anonymity for upload and download traffic, respectively. The thresholds $alpha\_upload$ and $alpha\_download$ govern the time-slot sampling process for upload and download directions, determining when traffic should be considered for transmission. The parameter $k$ specifies the anonymity set size, while $b$ denotes the tail padding length. Specifically, $b$ is used to monitor whether the transmission buffer remains empty over consecutive time slots, thereby determining whether dummy packet padding should be continued to preserve anonymity guarantees. 

In ALERT~\cite{qiao2024trace}, $min$ and $max$ control the min and max overhand thresholds. In GAPDiS~\cite{xie2025gapdis}, $max\_perturb$ denotes the maximum perturbation length. $max\_perturb=128$ means a maximum of 128 dummy packets are allowed to be inserted. In addition, $topk\_num$ denotes the number of top-$k$ candidate solutions retained during the search process, constrained within the range $[1, 5000 - M]$. The mutation rate, $muta\_rate$, controls the degree of stochastic variation introduced into candidate solutions. $tabu\_len\_multi$ determines the number of layers used to scale the tabu list length, thereby influencing the memory structure and the search procedure's diversification capability. In FRUGAL~\cite{wang2026cease}, $bwo\_para$ is the bandwidth overhead parameter, $limits$ is the number of training samples per class.

In Adaptive Tamaraw~\cite{khajavi2026lightening}, as in the Palette defense work, $k$ is the anonymity set size. $L$ is the Tamaraw bucket length parameter, Larger $L$ provides better security but higher overhead. $\alpha$ is the ECDIRE confidence threshold that controls when to switch from global to local parameters. In our design, the parameter website trace selection ratio $\rho$ defines the high intra-class diversity and low inter-class separability website trace selection ratio. The lower selection ratio will decrease the defense efficiency and lower the memory overhead for building the radix trie. Parameter $mutation$ controls the website morphing strategy. When $mutation = 1$, the model applies trace mutation; when $mutation = 0$, it performs full-trace morphing; and when $mutation = 3$, it uses trace cut-off morphing. For both training and evaluation, we set the number of epochs to $epochs = 50$. In the DAAE design, by default, 20\% of the dataset ($train\_ratio=0.2$) is used for model training, while the remaining 80\% is used for testing.

\begin{table}[t]
\renewcommand{\arraystretch}{1.2}
\caption{Parameter settings for evaluating WF defense models under an optimal trade-off between performance and overhead.}
\label{tab:hyperparameters}
\centering
\begin{tabular}{p{2.3cm}p{5.5cm}}
\toprule[0.75pt]
\rowcolor{olive!15}
\textbf{Defense models} & \textbf{Hyperparameters} \\
\midrule
FRONT            & $N_{c}=N_{s}=100$, $W_{min}=1s$, $W_{max}=8s$\\
\rowcolor{gray!25}
ALERT        & \makecell[l]{$min=0.1$, $max=0.5$} \\

Palette          & \makecell[l]{$u\_upload=u\_download=45$, $round=1$, \\ $alpha\_upload=alpha\_download=0.16$, \\ $k=30$, $b=20$} \\
\rowcolor{gray!25}
FRUGAL & $bwo\_para=0.8$, $limits=20$ \\

GAPDiS           & \makecell[l]{$max\_perturb=128$, $tabu\_len\_multi=5$, \\$topk\_num=10$, $muta\_rate=8$} \\
\rowcolor{gray!25}
Adaptive Tamaraw & \makecell[l]{$k=7$, $\alpha=0.9$, $L=1000$ } \\
\C (Ours)         & \makecell[l]{ $\rho=0.2$, $mutation=1$, $epochs=50$ } \\
\rowcolor{gray!25}
DAAE (Ours) & $train\_ratio=0.2$ \\

\bottomrule[0.75pt]
\end{tabular}
\vspace{-10pt}
\end{table}

\begin{figure*}[t]
    \centering
    \subfloat[Traffic trace patterns of website 1]{
    \includegraphics[width=0.31\linewidth]{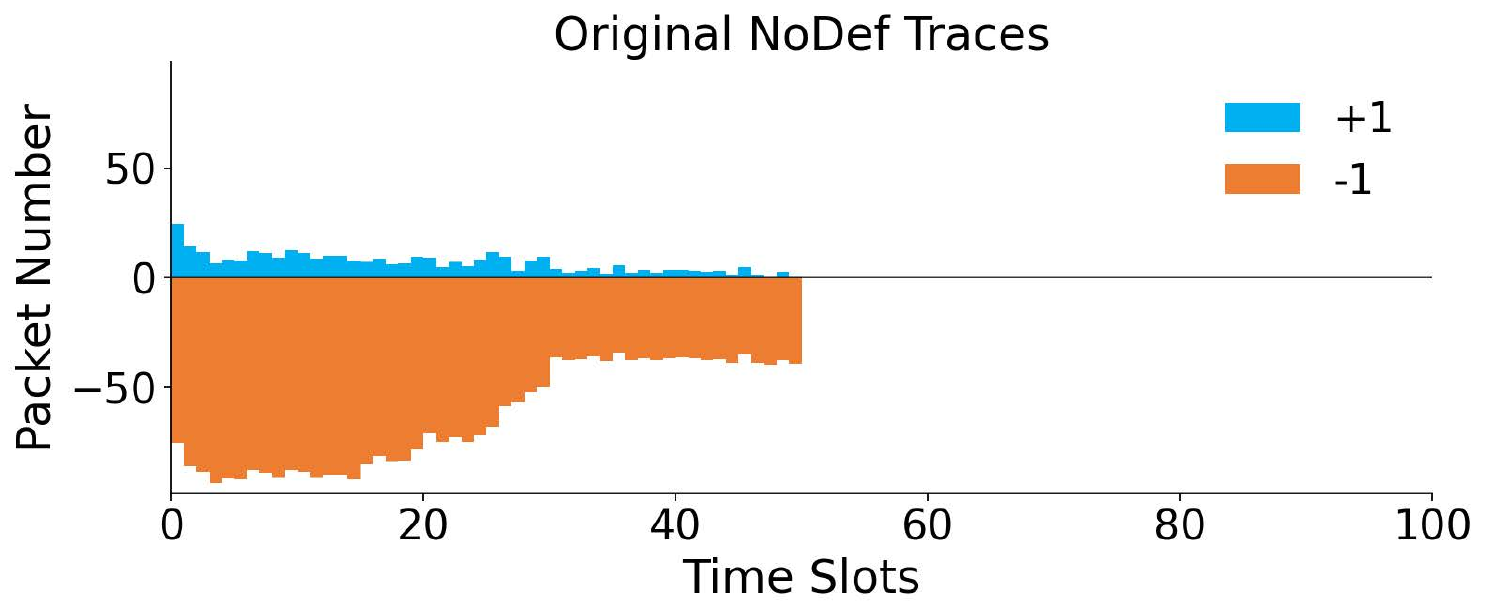}
    \includegraphics[width=0.31\linewidth]{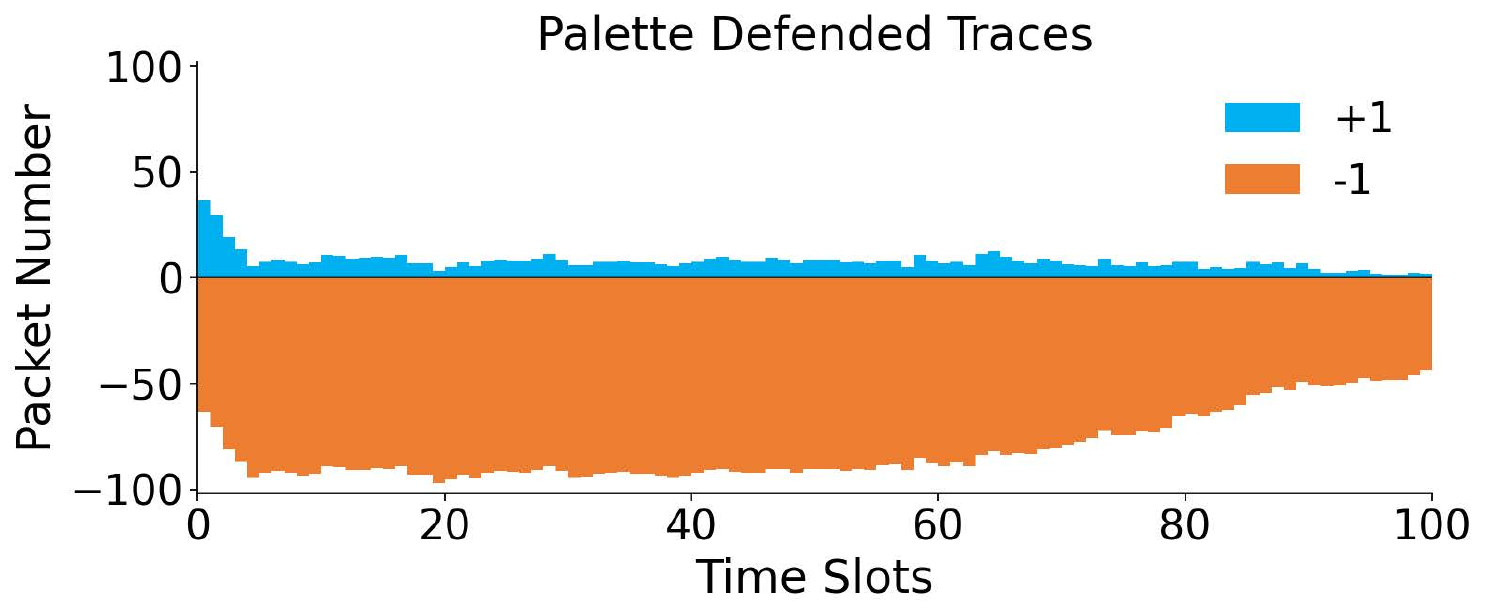}
    \includegraphics[width=0.31\linewidth]{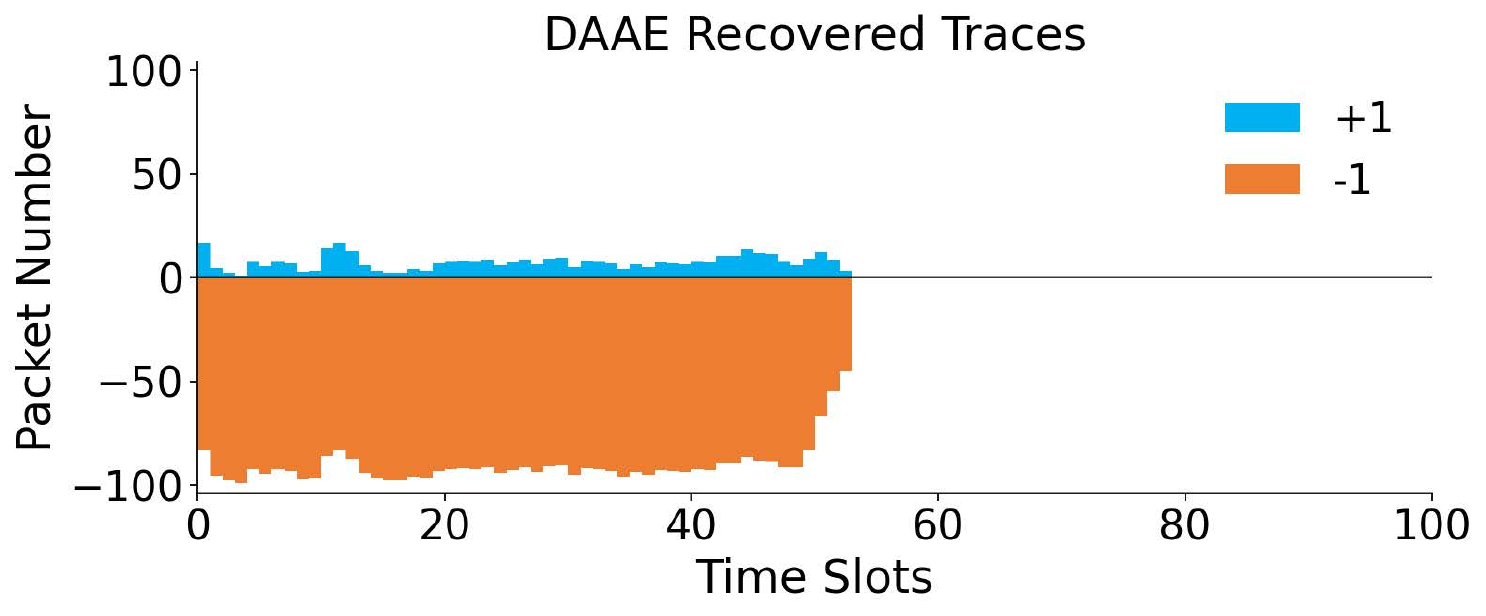}
    }

    \subfloat[Traffic trace patterns of website 2]{
    \includegraphics[width=0.31\linewidth]{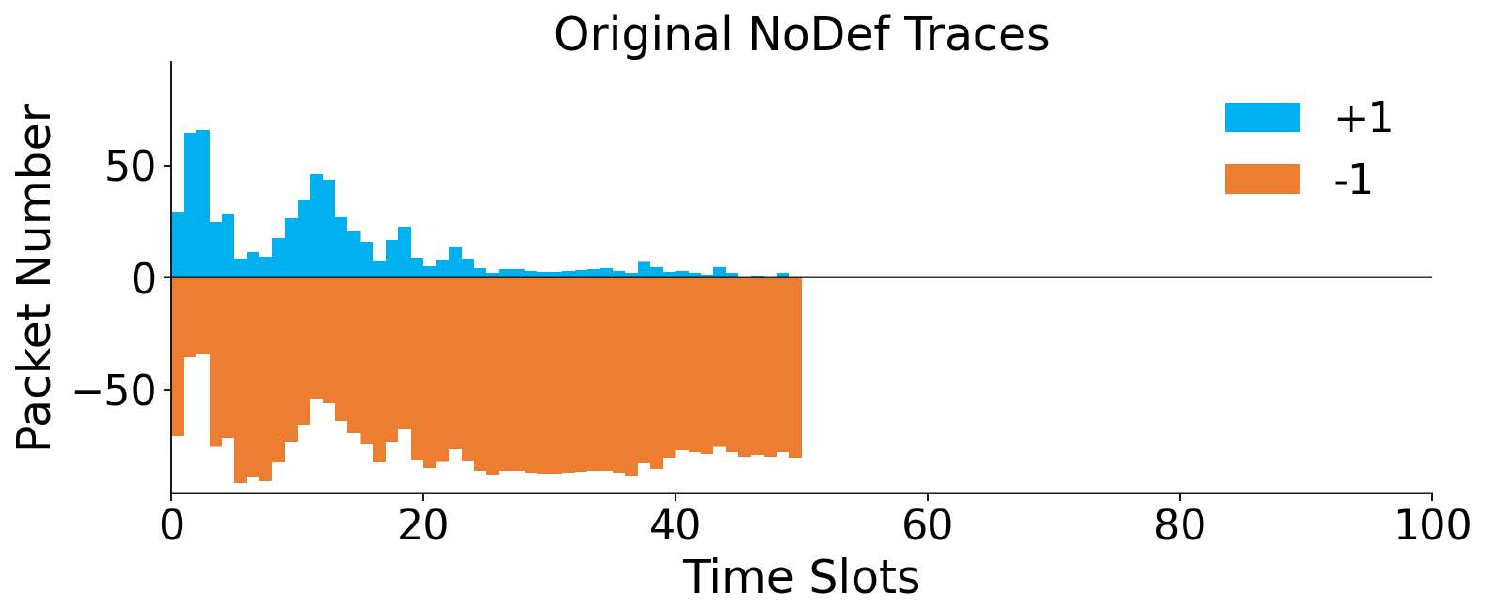}
    \includegraphics[width=0.31\linewidth]{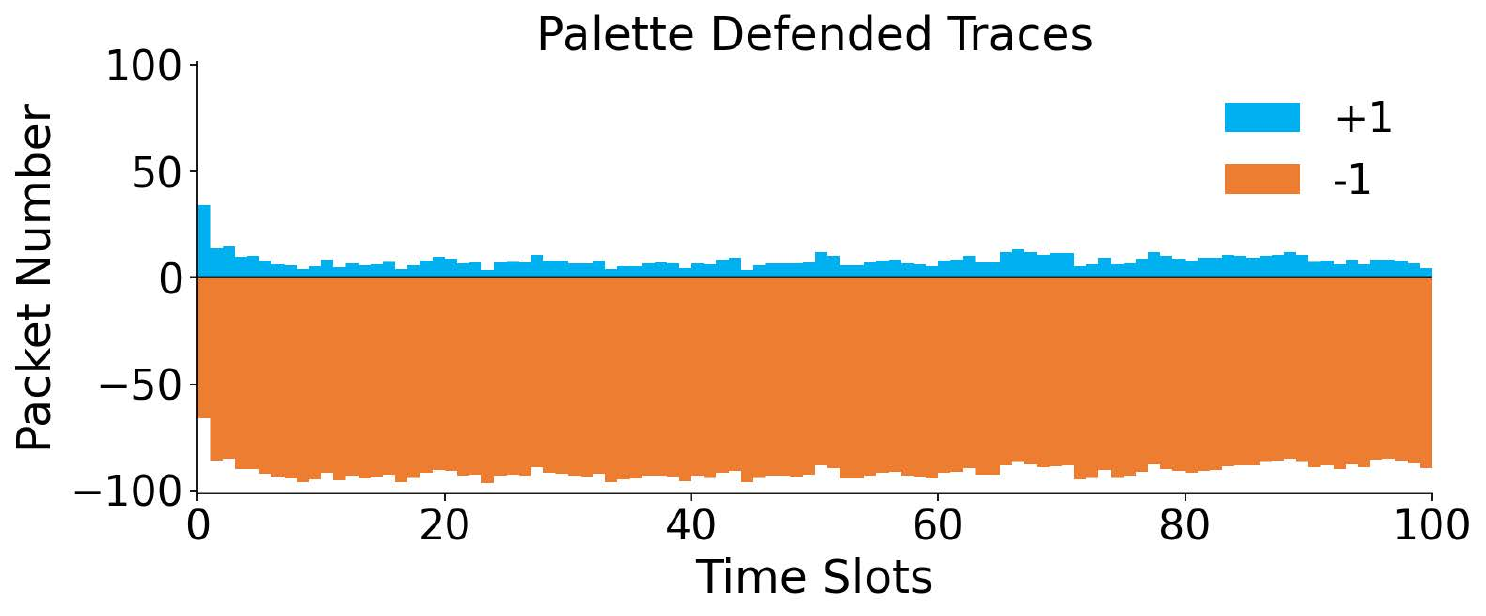}
    \includegraphics[width=0.31\linewidth]{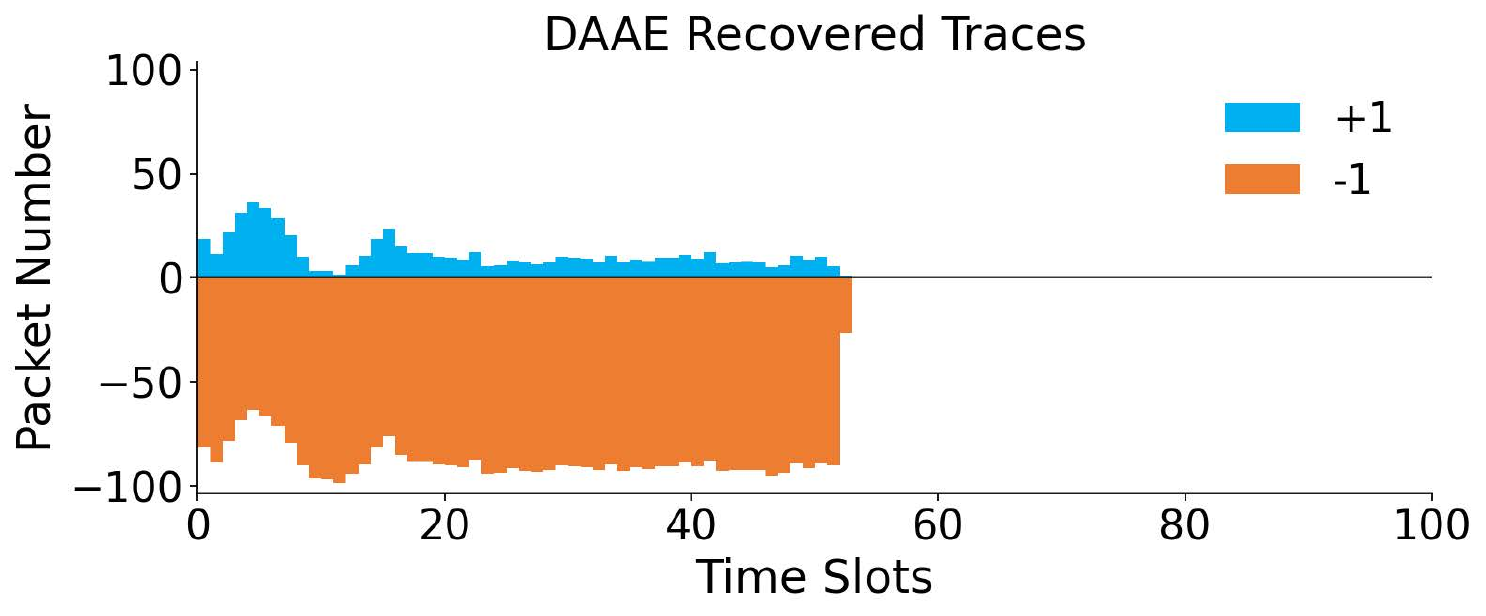}
    }

    \caption{Visualization of the TAMs for two websites under three settings: the original trace, the Palette~\cite{shen2024real}-defended trace, and the trace recovered by the DAAE-based attack under the GTT23 dataset~\cite{jansen2026measurement}. Each TAM covers the first 100 time slots, with each slot spanning 50ms. Outgoing and incoming packets are represented by blue (+1) and red ($-$1) bars, respectively, with bar height indicating the packet count within each slot. Each row presents the original, Palette-defended, and DAAE-recovered trace patterns from the same website.}
    \label{fig:motivation}
\end{figure*}

\begin{figure}[t]
  \centering
  \includegraphics[scale=0.40]{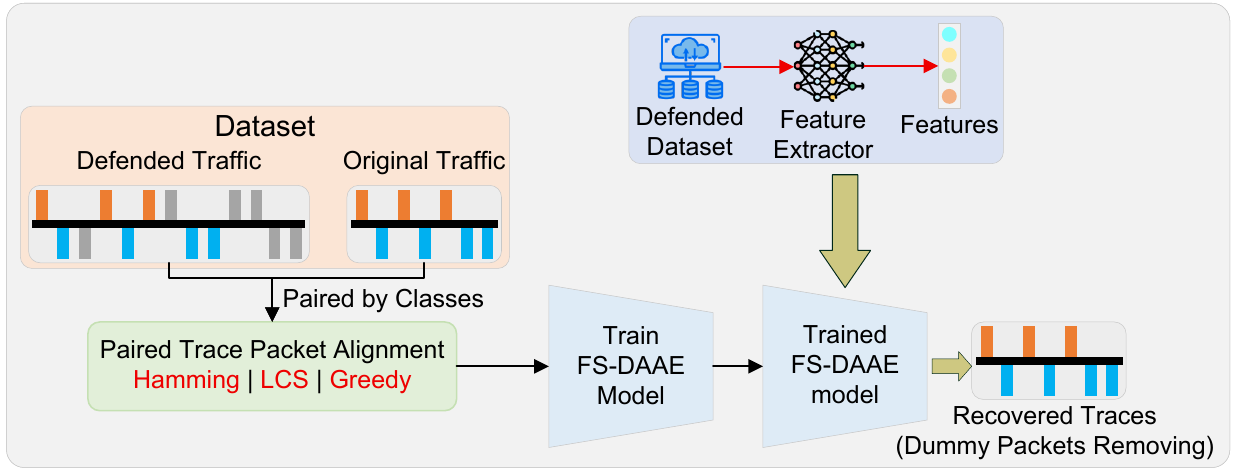}
  \caption{An overview of the DAAE architecture.}
  \label{fig:daae}
  \vspace{-0.2in}
\end{figure}

\subsection{Defense-Aware Autoencoder (DAAE) Attack}
\label{subsec:daae}

Through analyzing existing many-to-one defense mechanisms, we develop a new Defense-Aware Autoencoder (DAAE) attack that learns the characteristic packet-injection patterns introduced by WF defenses and removes the inferred dummy packets from unseen defended traces. 

\subsubsection{DAAE Formalization}
\label{subsec:daaeformula}
DAAE is a defense-aware autoencoder that learns the packet transformation
induced by a defense $\mathcal{D}$ from a small paired corpus of original and
defended traces. Let
\begin{equation}
  \mathcal{P}
  =
  \left\{
    \left(W_i,\, W_i^{'},\, y_i\right)
  \right\}_{i=1}^{\mathcal{R}},
  \qquad
  W_i^{'} = \mathcal{D}(W_i),
\end{equation}
where $y_i\in\mathcal{Y}$ denotes the website label and $i$ indexes the
$i$-th trace pair.

For each pair $(W_i, W_i^{'}, y_i)\in\mathcal{P}$, sequence alignment is
performed to derive a supervision mask without requiring white-box access to
the defense mechanism $\mathcal{D}$:
\begin{equation}
  \mathbf{m}
  =
  \mathrm{Align}\bigl(W_i, W_i^{'}\bigr).
\end{equation}
Here, $\mathrm{Align}$ can employ different sequence alignment strategies,
including Hamming distance, Longest Common Subsequence (LCS)~\cite{paterson1994longest}, and a greedy method.

A defense-agnostic encoder $\Phi$ maps each defended trace to a
fixed-width tensor with $c$ feature channels, corresponding to packet
direction, min-max normalized timestamps, and log inter-packet gaps:
\begin{equation}
  X_{i}^{'} = \Phi(W_{i}^{'}) \in \mathbb{R}^{c \times L},
  \label{eq:features}
\end{equation}
where $\mathbb{R}^{c \times L}$ denotes the space of real-valued matrices
with $c$ feature channels and $L$ trace positions.

The defense-aware autoencoder $f_{\theta}$ takes the defended feature tensor
$X^{'}$ together with same-class support pairs
$\mathcal{S}=(X^{'},\mathbf{m})$ and learns to \textit{infer the transformation
pattern of the defense}. Given the defense-aware training set
$\mathcal{P}_{\mathrm{tr}}$, the model is trained by minimizing the
discrepancy between its predictions and the alignment-derived supervision
masks:
\begin{equation}
  \theta^{*}
  =
  \arg\min_{\theta}
  \;
  \mathbb{E}_{(W,W^{'})\sim\mathcal{P}_{\mathrm{tr}}}
  \left[
    \mathcal{L}
    \bigl(
      \theta;\,
      W, W^{'}, \mathbf{m}, \mathcal{S}
    \bigr)
  \right].
\end{equation}

As shown in \figureautorefname~\ref{fig:daae}, given a query trace and a set of support traces, a 1D CNN extracts their representations, while a prototype-based Feature-wise Linear Modulation (FiLM)~\cite{perez2018film} layer and cross-attention enable the query to incorporate defense-specific injection patterns learned from the supports. The DAAE model employs two complementary prediction paths: a convolutional head for identifying local injection boundaries and a U-Net~\cite{ronneberger2015u} with dilated convolutions for capturing longer dummy-packet runs.

DAAE is trained using paired original and defended traces. For each pair, we first align the two traces to identify the positions at which dummy packets are inserted. We implement three alignment strategies: LCS, Hamming distance, and a greedy method, with LCS used by default. The resulting alignment provides dummy-packet positions under different traffic patterns, which are then used as supervision for training DAAE. 
To optimize the DAAE model, we apply focal~\cite{lin2017focal} and Dice~\cite{sudre2017generalised} losses. We observe that dummy-injection defenses typically add many padding packets. A plain per-packet $L_2$ loss can ignore the rarer class, while focal loss up-weights hard decisions. Meanwhile, Dice loss measures overlap between the predicted and labeled keep sets over the whole trace. We also apply a keep-rate regularizer so the predicted fraction of original packets matches the labeled fraction. 

We further develop a few-shot variant~\cite{wang2020generalizing}, FS-DAAE. To match a real WF setting, the attacker does not know the class of a captured trace. Only a limited paired split (20\% of instances per class by default) is labeled and used to train DAAE. The remaining instances are never paired with originals at test and validation phases. At inference, the attacker observes only a captured defended trace and has no access to its original trace, true class label, or victim-side state. Neither DAAE nor FS-DAAE uses the true class labels of validation or test traces during recovery. Instead, the trained model directly infers the defense's dummy-packet injection pattern and removes the predicted dummy packets from each defended trace. 


\subsubsection{DAAE Decoding Stategy}
At test time, DAAE does not have access to the original trace. Instead, it maintains a small support pool constructed from the training set and applies the same inference pipeline to the captured defended traffic.
We designed two strategies to decode the defended packets into the original ones.
We find that directly thresholding at 0.5 for the predicted keep probabilities can incorrectly remove original packets or retain dummy packets, disrupting the original packet-direction subsequence. Therefore, we perform a sequence-level decoding step against support prototypes, where each prototype represents the packet-direction sequence of the packets labeled as original in a support trace. We introduce a skip bias to control the trade-off between retaining and removing packets and select the bias that produces a keep rate closest to the median keep rate of the support set. When multiple supports are available, we first consider supports whose kept-trace lengths are closest to the target length and select the candidate yielding the best sequence-level decoding score. The recovered trace is then constructed by retaining the captured packets selected by the final mask. Finally, we train a standard WF classifier, such as DF and RF, on the recovered traces and evaluate its classification performance to quantify how effectively DAAE restores the fingerprint information removed by the defense.

As shown in the third column of \figureautorefname~\ref{fig:motivation}, the recovered traces do not exactly match the originals, but much of their original traffic structure is restored. This result suggests that defense-induced traffic artifacts can be exploited to partially reverse the effect of traffic regularization, potentially undermining the defense against a sufficiently adaptive adversary. \figureautorefname~\ref{fig:motivation} demonstrates that the original non-defended traces show distinct patterns in their Traffic Aggregation Matrices (TAMs)~\cite{shen2023subverting}. Palette~\cite{shen2024real} aims to reduce this variability by regularizing traffic into a common pattern. However, the resulting defended traces can still preserve noticeable differences. More importantly, the regularization process introduces systematic dummy-packet injection patterns that may themselves become learnable by an adaptive adversary. 

\subsection{Row-wise z-score Normalization}
\label{subsec:zscore}

\begin{algorithm}[t]
\caption{Row-wise $z$-score normalization of matrix $\mathbf{X}$}
\label{alg:rowzscore}
\SetAlgoLined
\KwIn{Matrix $\mathbf{X} \in \mathbb{R}^{N \times L}$, stability constant $\varepsilon > 0$}
\KwOut{Row-normalized matrix $\mathbf{Z} \in \mathbb{R}^{N \times L}$}

\For{$i \gets 1$ \KwTo $N$}{
    $\mu_i \gets \frac{1}{L}\sum_{j=1}^{L} X_{ij}$ \tcp{Compute row mean}
    
    $\sigma_i \gets \sqrt{\frac{1}{L}\sum_{j=1}^{L}(X_{ij}-\mu_i)^2}$ \tcp{Compute row standard deviation}
    
    \For{$j \gets 1$ \KwTo $L$}{
        $Z_{ij} \gets \frac{X_{ij}-\mu_i}{\sigma_i + \varepsilon}$ \tcp{Row-wise $z$-score}
    }
}
\Return $\mathbf{Z}$
\end{algorithm}

Row-wise $z$-score normalization applies a separate mapping standardization to each row of a matrix $\mathbf{X}\in\mathbb{R}^{N\times L}$.$N$ is the number of traces and $L$ the fixed prefix length used for comparison. This algorithm treats each row as one feature vector. As shown in \algorithmcfname~\ref{alg:rowzscore}, for each row $i$, the algorithm computes the row mean $\mu_{i}$ and a row scatter scale $\sigma_{i}$ from the entries $X_{ij},{j=1\to L}$, then sets $Z_{ij}=(X_{ij}-\mu_{i})/(\sigma_{i}+\varepsilon)$. A stability constant $\varepsilon>0$ is added to the denominator to avoid division by zero when a row is constant or near constant. In implementation, $\sigma_{i}$ follows either standard deviation (divide by $L$ under the square root) or sample standard deviation (divide by $L{-}1$). $\varepsilon$ and the precise $\sigma_{i}$ are the key parameters that materially affect the normalized rows.

This method will centralize each trace around zero and rescale the centralization spread. So that rows with different absolute magnitudes or local variance become more comparable before similarity measures are computed. The operation is performed independently on each row, ensuring that no statistical information is shared between traces. Consequently, the identity of each sample is preserved while eliminating variations that arise solely from scale differences along the sequence dimension.

\subsection{Defense Performance Against WF Attacks Without Adversarial Training}
\label{subsec:withoutadvent}

\begin{table*}[t]
\centering
\small
\setlength{\tabcolsep}{6pt}
\renewcommand{\arraystretch}{1.2}

\caption{Evaluation of defense effectiveness without AT-based attack across six SOTA WF defenses using DF~\cite{sirinam2018deep} and DS-19~\cite{gong2020zero} datasets under four SOTA WF attack models in the closed-world scenario.}
\label{tab:withoutadvtrain}

\scalebox{0.83}{
\begin{tabular}{lcccccccc}
\toprule
\rowcolor{olive!15}
\textbf{Dataset}                 & \multicolumn{4}{c}{DF~\cite{sirinam2018deep}}            & \multicolumn{4}{c}{DS-19~\cite{gong2020zero}}         \\
\multirow{2}{*}{\diagbox{\textbf{Defenses}}{\textbf{Attacks}}} & DF~\cite{sirinam2018deep}  & Var-CNN~\cite{bhat2018var}  & RF~\cite{shen2023subverting} & NetCLR~\cite{bahramali2023realistic} & DF~\cite{sirinam2018deep}  & Var-CNN~\cite{bhat2018var}  & RF~\cite{shen2023subverting} & NetCLR~\cite{bahramali2023realistic} \\
\cmidrule(r){2-5} \cmidrule(r){6-9}
& \multicolumn{4}{c}{Accuracy (\%)} & \multicolumn{4}{c}{Accuracy (\%)} \\
\midrule
NoDef & 98.33 & 98.83 & 98.93 & 98.53 & 97.23 & 96.75 & 90.94 & 97.50   \\

\rowcolor{gray!25}
FRONT~\cite{gong2020zero} & 27.54
& 28.71
& 35.18
& 25.79
& 37.23
& 46.75
& 40.94
& 47.53 \\


ALERT~\cite{qiao2024trace} & 17.48
& 24.07
& 29.85
& 19.82
& 28.80
& 36.50
& 37.80
& 29.61 \\

\rowcolor{gray!25}
Palette~\cite{shen2024real} & 2.15
& 2.49
& 3.01
& 1.92
& 1.10
& 1.75
& 2.23
& 1.01 \\

GAPDiS~\cite{xie2025gapdis} & 7.19
& 4.60
& 11.79
& 10.74
& 1.65
& 2.89
& 1.50
& 2.88  \\

\rowcolor{gray!25}
FRUGAL~\cite{wang2026cease} & 5.29
& 5.75
& 11.17
& 6.73
& 5.10
& 5.63
& 11.26
& 6.29 \\

Adaptive Tamaraw~\cite{khajavi2026lightening} & 5.05
& 5.15
& 6.28
& 5.23
& 1.23
& 1.77
& 1.82
& 1.00 \\

\rowcolor{gray!25}
\textbf{\C (Lightweight)} & \textbf{1.77}
& \textbf{1.85}
& \textbf{1.55}
& \textbf{1.62}
& \textbf{0.17}
& \textbf{0.11}
& \textbf{0.30}
& \textbf{0.95} \\

\C (Medium) & 3.39
& 3.00
& 3.19
& 3.40
& 2.37
& 2.30
& 2.33
& 2.40 \\

\rowcolor{gray!25}
\C (Heavy) & 2.88
& 2.75
& 3.18
& 2.92
& 2.70
& 2.61
& 2.78
& 2.80 \\

\bottomrule
\end{tabular}
}
\vspace{-0.2in}
\end{table*}

We compare the proposed \C defense with six defenses across the DF and DS-19 datasets under four representative WF attack models. As shown in \tableautorefname~\ref{tab:withoutadvtrain}, \C consistently achieves the lowest WF attack accuracy.

Specifically, on the DF dataset, all variants of \C consistently outperform prior defenses by a substantial margin. In particular, the lightweight variant of \C reduces attack accuracy to as low as 1.55\% across all attack models, yielding a clear improvement over the strongest existing baseline. While \textit{Palette} achieves a minimum accuracy of 1.92\%, \C further advances the SOTA defense by achieving lower and more consistent accuracy across attacks. Compared to other competitive defenses, such as \textit{Adaptive Tamaraw} and \textit{GAPDiS}, \C demonstrates a marked reduction in attack success rates, underscoring its robustness against diverse attack architectures. 
Importantly, the medium and heavy configurations of \C also exhibit strong performance, consistently surpassing most baselines and indicating that the effectiveness of \C is maintained across different resource-performance trade-offs. The advantage of \C becomes even more pronounced on the more challenging DS-19 dataset. Here, lightweight \C achieves near-zero attack accuracy, substantially surpassing all prior defenses, including \textit{Palette} and \textit{Adaptive Tamaraw}, whose best results remain above 1.00\%. 

Moreover, while existing defenses exhibit noticeable performance degradation or variability across different attacks, \C maintains consistently low accuracy across all models, highlighting its strong robustness capability. The proposed trace mutation mechanism alters the underlying traffic patterns relative to the original website trace and target-morphing traces, thereby reducing their distinguishability under WF attacks. By mitigating such pattern redundancy, \C achieves more diverse and less predictable traces, contributing to its robust and stable defense performance.

\section{Pluggable Transporting Deployment}
\label{subsec:pluggableent}

\C is a trace-driven pluggable transport that morphs packet timing and directionality to mimic realistic encrypted web traffic while preserving Tor’s deployment model. The design sits on top of an authenticated, encrypted transport channel and therefore treats morphology as a traffic-shaping problem rather than a cryptographic one: confidentiality and integrity are provided by the underlying channel, while \C controls observable metadata. At a high level, \C continuously emits packets according to a selected reference trace template and maps application payload into template slots designated for outbound transmission. When no payload is available, but the defense is active, it emits dummy packets to preserve schedule consistency, reducing leakage from idle periods and burst structure.

The core of \C is a prefix-conditioned template engine. An offline dataset of representative traces is converted into directional sequences and indexed in a radix trie. Each online connection maintains a rolling prefix of recent packet directions and queries the trie for templates that share the same prefix. Additionally, one candidate is randomly selected and may be mutated over time to enforce feasibility constraints. This gives \C two properties that are hard to obtain simultaneously with static schedules: (i) continuity with recent flow evolution, and (ii) controlled stochasticity that avoids deterministic fingerprints. In effect, the transport performs constrained sampling over a realistic trace manifold rather than replaying fixed traces.

\subsection{Effectiveness of Prefix Length in Radix Trie}
\label{subsec:appradix}

\begin{figure}[t]
    \centering
    \subfloat[Effectiveness of morphing trace targeting on the DF dataset.\label{subfig:radixdf}]{
        \includegraphics[width=0.45\linewidth]{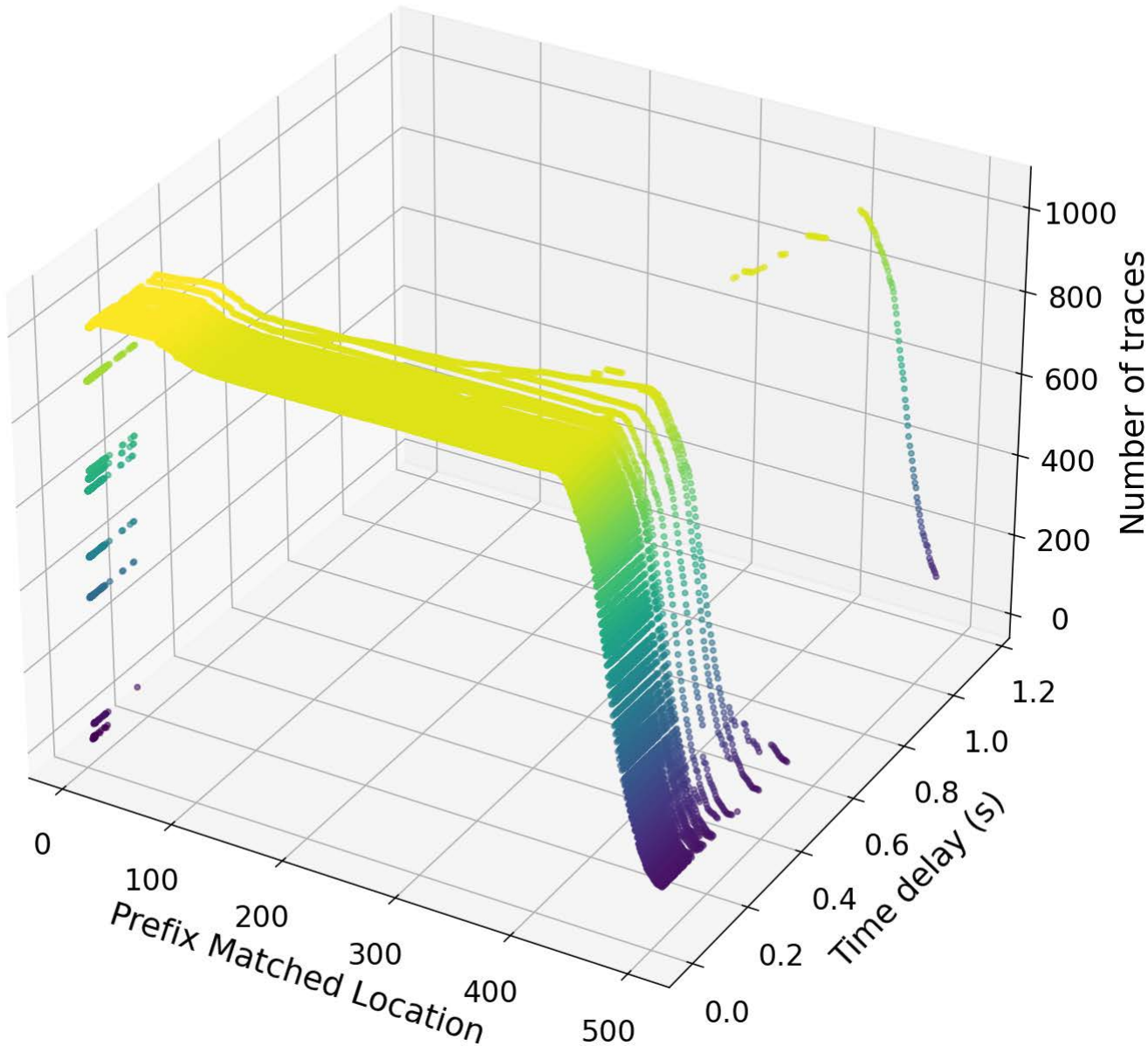}
    }
    \hfill
    \subfloat[Effectiveness of morphing trace targeting on the DS-19 dataset.\label{subfig:radixdf19}]{
        \includegraphics[width=0.45\linewidth]{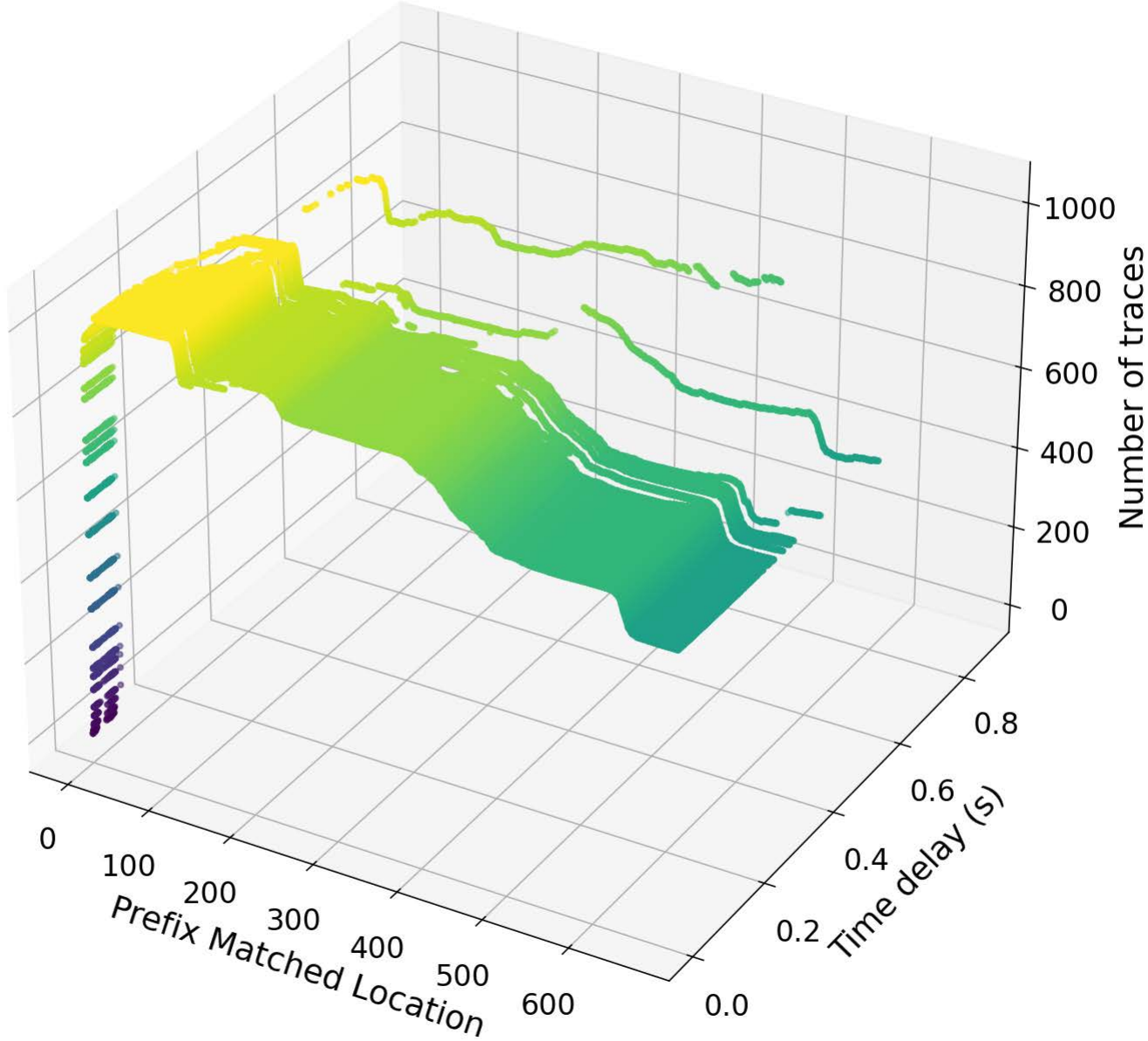}
    }
    \caption{Overview of the effectiveness of morphing trace targeting through radix trie on the DF and DS-19 datasets.}
    \label{fig:radixperf}
\end{figure}

\figureautorefname~\ref{fig:radixperf} illustrates the efficacy of morphing trace targeting using a radix trie on the DF and DS-19 datasets. The figures show the targeting latency, the depth at which the morphing target is determined, and the number of candidate website traces sharing the same prefix. Specifically, \figureautorefname~\ref{subfig:radixdf} presents the prefix-matching performance on the DF dataset. The results indicate substantial intra-dataset diversity. Most website traces exhibit highly discriminative prefixes after about 450 packet transmissions. This property enables efficient target selection, with the delay of completing morphing-trace targeting within 0.8 seconds in the majority of website traces. In contrast, \figureautorefname~\ref{subfig:radixdf19} shows that the DS-19 dataset exhibits lower diversity. A few website traces continue to share identical prefixes even beyond the truncated radix trie length threshold. In such cases, the system resorts to random selection among the remaining candidates. Nevertheless, the overall targeting delay remains below 1 second.

\subsection{Information Leakage Analysis}
\label{subsec:theory}

We provide an idealized analysis to characterize the privacy benefit of randomized many-to-many trace morphing. This analysis is intended to explain the intuition behind the design rather than provide a security guarantee for the complete implementation. In particular, it abstracts away timing, online prefix-based candidate selection, trace mutation, and other implementation-dependent behavior, whose privacy impact is evaluated empirically.

\paragraph{Idealized Morphing Model.}
Let $W$ denote the website class and $C$ denote the candidate morphing trace observed by an adversary. For each website $w$, let $\mathcal{C}_{w}$ denote the set of candidate traces to which its traffic may be mapped. We consider an idealized setting in which the defense randomly selects a candidate from $\mathcal{C}_{w}$ and the observable output is determined only by the selected candidate. Thus, the information available to the adversary is characterized by $I(W;C)$.

A deterministic mapping (\eg, one-to-one) preserves the website identity through the selected candidate and therefore does not protect in this idealized model. In contrast, the many-to-many mapping of \C allows candidate traces to be shared across different websites. For an observed candidate $c$, define
\begin{equation}
\mathcal{W}(c)={w : P(C=c\mid W=w)>0},
\end{equation}
which contains all website classes that may produce $c$. If $C=c$ is observed, the adversary must distinguish among the classes in $\mathcal{W}(c)$ rather than infer a unique source website.

Suppose that candidate $c$ can be produced by $m$ website classes with equal priors and equal conditional selection probabilities. By Bayes' rule,
\begin{equation}
P(W=w\mid C=c)=\frac{1}{m}, \qquad w\in\mathcal{W}(c),
\end{equation}
and therefore
\begin{equation}
H(W\mid C=c)=\log_{2} m.
\end{equation}
Under these assumptions, increasing the number of source classes sharing a candidate increases the residual uncertainty of an adversary in this idealized setting. More generally, the leakage is
\begin{equation}
I(W;C)=H(W)-H(W\mid C),
\end{equation}
so candidate sharing can reduce leakage when it increases $H(W\mid C)$.

\paragraph{Scope of the Analysis.}
The above result applies only to the idealized candidate-mapping abstraction and should not be interpreted as a security bound for the complete \C system. In the implementation, candidate selection is conditioned on prefixes observed from real-time traffic, and the final trace may be further modified by timing behavior, dummy-cell insertion, trace cut-off, and mutation. In particular, trace mutation is designed to reduce directional mismatch and transmission overhead; reducing this mismatch does not necessarily reduce information leakage and may, in some cases, preserve additional characteristics of the original traffic. We therefore make no information-theoretic claim that mutation decreases $I(W;Y)$ or the divergence between website-specific output distributions. Instead, the privacy of the complete system, including these dependencies, is quantified empirically using WeFDE~\cite{li2018measuring} leakage measurements and evaluated against the WF attacks considered in Section~\ref{subsec:leakage}.

\begin{table*}[t]
\centering
\small
\setlength{\tabcolsep}{3pt}
\renewcommand{\arraystretch}{1.2}

\caption{Information leakage for each feature category (Bit).}
\label{tab:infoleakage}
\scalebox{0.83}{
\begin{tabular}{lllccclccc}
\toprule
\rowcolor{olive!15}
Categories & Pkt. Count & Time & $Interval_{Mean}$ & $Interval_{std}$ & Pkt. Distribution & Burst & First 20 & First 30 & Diag GNB \\

\midrule
\rowcolor{gray!25}
NoDef& 1.231
& 0.818
& 0.828
& 0.660
& 1.233
& 1.030
& 0.675
& 0.520
& 5.82 \\

FRONT& 1.253 {\footnotesize{\color{red}{$\uparrow$0.022}}}
& 0.786 
& 0.753
& 0.459
& 0.795
& 0.157
& 0.167
& 0.479
& 5.57 \\

\rowcolor{gray!25}
RegulaTor & 1.183
& 0.972 {\footnotesize{\color{red}{$\uparrow$0.154}}}
& 0.662
& 0.288
& 0.687
& 1.251 {\footnotesize{\color{red}{$\uparrow$0.221}}}
& 0.067
& 0.044
& 5.49 \\

Palette & 1.059
& 0.788
& 0.612
& 0.328
& 1.043
& 0.431
& 0.025
& 0.036
& 4.62 \\

\rowcolor{gray!25}
\C & 0.727
& 0.709
& 0.614
& 0.495
& 0.718
& 0.374
& 0.351
& 0.464
& 4.57 \\

\bottomrule
\end{tabular}
}
\vspace{-0.2in}
\end{table*}

\subsection{Information Leakage Measurement}
\label{subsec:leakage}

As \C employs randomized many-to-many traffic morphing, multiple browsing sessions may produce identical or highly similar feature representations. We retain these repeated observations when computing information leakage because they are part of the empirical output distribution generated by the defense. Removing duplicated feature vectors would alter their empirical frequencies and may therefore change the estimated relationship between website labels and observable traffic features. We consequently apply WeFDE~\cite{li2018measuring} to the complete defended trace set without deduplicating feature representations. As a robustness check, we additionally evaluate a deduplicated variant and report the resulting leakage separately to determine whether the conclusions are sensitive to duplicate handling.

\tableautorefname~\ref{tab:infoleakage} reports the information leakage estimated from the complete empirical trace distribution across eight representative feature categories, together with the aggregate Diagonal Gaussian Naïve Bayes (Diag GNB) measurement. Lower leakage indicates that the observable traffic features reveal less information about the corresponding website labels. \C achieves lower or comparable leakage across most feature categories and obtains the lowest aggregate Diag GNB leakage among the evaluated defenses. \C reduces packet-count leakage to 0.727, compared to 1.231 for NoDef and 1.253 for FRONT, and lowers packet-distribution leakage to 0.718, compared to 1.233 for NoDef. In contrast, RegulaTor exhibits higher leakage for timing-related features, with a mean interval of 0.972 and a burst of 1.251, whereas \C maintains lower, more stable values across these dimensions. Palette achieves lower leakage in certain fine-grained features, such as First 20 (packets) with 0.025 and First 30 (packets) with 0.036, yet its leakage remains higher than \C in several other categories. In addition, \C attains the lowest aggregate leakage under Diag GNB with a value of 4.57, outperforming NoDef with 5.82, FRONT with 5.57, and RegulaTor with 5.49, indicating more comprehensive protection across diverse feature spaces.

\subsection{Security Analysis of the Idealized Many-to-Many Mapping}
\label{subsec:boundproof}

Let $W$ denote the set of original traces and $W'$ the set of defended traces. \C first extracts a subset $C \subseteq W$ of candidate traces exhibiting high intra-class diversity and low inter-class separability. The defense defines a many-to-many relation $\mathcal{R} \subseteq W \times C$, where each trace $w \in W$ can be morphed into one or multiple candidate traces in $C$. For any $c \in C$, define the pre-image:
\begin{equation}
\mathcal{R}^{-1}(c) = \{ w \in W : (w, c) \in \mathcal{R} \}.
\end{equation}
Let $\mathcal{R}^{-1}_{y}(c)$ denote the subset of traces in $\mathcal{R}^{-1}(c)$ belonging to class $y$.

\textbf{Website Trace Injection Upper Bound.}
We define the effective injection upper bound for candidate trace $c$ as:
\begin{equation}
\tilde{\delta}(c) = 
\frac{|\mathcal{R}^{-1}(c)|}
{\max_{y} |\mathcal{R}^{-1}_{y}(c)|}.
\end{equation}
This reflects both the number of traces mapped to $c$ and their class diversity.

\begin{lemma}
For any candidate trace $c \in C$, the attacker’s success probability is bounded by:
\begin{equation}
\Pr[\text{success} \mid c] \le \frac{1}{\tilde{\delta}(c)}.
\end{equation}
\end{lemma}

\begin{proof}
Under \C, all traces in $\mathcal{R}^{-1}(c)$ are morphed into the same observable pattern $c$ and are thus indistinguishable to the adversary. The optimal strategy is to predict the majority class within $\mathcal{R}^{-1}(c)$, yielding:
\begin{equation}
\Pr[\text{success} \mid c] 
= \max_y \frac{|\mathcal{R}^{-1}_y(c)|}{|\mathcal{R}^{-1}(c)|}
= \frac{1}{\tilde{\delta}(c)}.
\end{equation}
\end{proof}

\begin{theorem}
Let the defended trace distribution be induced by sampling candidate traces $c \in C$. Define:
\begin{equation}
\frac{1}{\delta} = \mathbb{E}_{c \sim C} \left[ \frac{1}{\tilde{\delta}(c)} \right].
\end{equation}
Then the attacker’s average success probability is bounded by:
\begin{equation}
\Pr[\text{success}] \le \frac{1}{\delta}.
\end{equation}
\end{theorem}

\begin{proof}
From the lemma, for any observed trace $c$:
\begin{equation}
\Pr[\text{success} \mid c] = \frac{1}{\tilde{\delta}(c)}.
\end{equation}
Taking expectation over the distribution of candidate traces:
\begin{equation}
\Pr[\text{success}] 
= \mathbb{E}_{c} \left[ \Pr[\text{success} \mid c] \right]
= \mathbb{E}_{c} \left[ \frac{1}{\tilde{\delta}(c)} \right]
\le \frac{1}{\delta}.
\end{equation}
\end{proof}

Unlike mapping traces to fixed or highly regularized patterns, \C leverages a candidate set $C$ with high intra-class diversity and low inter-class separability. By allowing each trace to map to multiple candidates, the pre-image $\mathcal{R}^{-1}(c)$ becomes both large and label-diverse, increasing $\tilde{\delta}(c)$ and reducing attacker success probability, while avoiding the overhead of strict normalization.

\end{document}